%% file: tacas23-main.tex
\documentclass[runningheads]{llncs}
\usepackage[T1]{fontenc}
\usepackage{graphicx}
\makeatletter
\RequirePackage[bookmarks,unicode,colorlinks=true]{hyperref}%
\def\@citecolor{blue}%
\def\@urlcolor{blue}%
\def\@linkcolor{blue}%

\def\orcidID#1{\smash{\href{http://orcid.org/#1}{\protect\raisebox{-1.25pt}{\protect\includegraphics{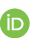}}}}}
\makeatother

\newtheorem{contribution}{Contribution}

\usepackage{bbding} 
\usepackage[table]{xcolor}
\definecolor{lightgray}{gray}{0.95}
\usepackage[utf8]{inputenc}
\usepackage{amsmath}
\usepackage{amssymb}
\usepackage{enumitem}
\usepackage{tikz}
\usetikzlibrary{automata,positioning,calc}
\usetikzlibrary{decorations.pathreplacing}
\usepackage[linesnumbered,ruled,vlined]{algorithm2e}
\usepackage{soul}
\usepackage{csquotes}
\usepackage{pgfplots} 
\usepackage{nicefrac}
\usepackage{thm-restate}
\usepackage{booktabs}
\usepackage{adjustbox}
\usepackage{pifont}
\newcommand{\cmark}{\ding{51}}%
\newcommand{\xmark}{\ding{55}}%
\usepackage{caption}
\usepackage{subcaption}
\usepackage{todonotes}

\pgfplotsset{compat=1.15}
\usepackage{mathrsfs}
\usetikzlibrary{arrows} 
\usepackage{pgfplots}
\pgfplotsset{compat=newest}
\usepgfplotslibrary{fillbetween}
\usepackage{cleveref}

\input{macros} 

\makeatletter
\def\thanks#1{\protected@xdef\@thanks{\@thanks
        \protect\footnotetext{#1}}}
\makeatother

\usepackage{etoolbox}
\newtoggle{arxiv}
\toggletrue{arxiv}

\begin{document}
\title{
    Certificates for Probabilistic Pushdown Automata via Optimistic Value Iteration
    \thanks{\setlength{\leftskip}{0em}%
        This work is supported by the DFG research training group 2236 UnRAVeL, the ERC advanced research grant 787914 FRAPPANT, and the European Union’s Horizon 2020 research and innovation programme under the Marie Skłodowska-Curie grant agreement No 101008233 (MISSION).}
}
\titlerunning{Certificates for Probabilistic Pushdown Automata via OVI}
%
\author{
    Tobias Winkler$^{\text{\Envelope}}$%
    \orcidID{0000-0003-1084-6408}
    \and
    Joost-Pieter Katoen%
    \orcidID{0000-0002-6143-1926}
}
\authorrunning{T. Winkler and J.-P. Katoen}
%
\institute{
    RWTH Aachen University, Aachen, Germany\\
    \texttt{\{tobias.winkler,katoen\}@cs.rwth-aachen.de}
}
\maketitle              
\begin{abstract}
    Probabilistic pushdown automata (pPDA) are a standard model for discrete probabilistic programs with procedures and recursion.
    In pPDA, many quantitative properties are characterized as least fixpoints of polynomial equation systems.
    In this paper, we study the problem of \emph{certifying} that these quantities lie within certain bounds.
    To this end, we first characterize the polynomial systems that admit easy-to-check certificates for validating bounds on their least fixpoint.
    Second, we present a sound and complete Optimistic Value Iteration algorithm for computing such certificates.
    Third, we show how certificates for polynomial systems can be transferred to certificates for various quantitative pPDA properties.
    Experiments demonstrate that our algorithm computes succinct certificates for several intricate example programs as well as stochastic context-free grammars with $> 10^4$ production rules.

\keywords{Probabilistic Pushdown Automata  \and Probabilistic Model Checking \and Certified Algorithms \and Probabilistic Recursive Programs.}
\end{abstract}

\input{1_intro}

\input{2_ppsovi}
\input{5_ppda}

\input{5_implementation}
\input{6_conclusion}

\subsubsection*{Data availability statement}
The datasets generated during and/or analysed during the current study are available in the Zenodo repository~\cite{zenodo}.

%
%
%
\bibliographystyle{splncs04}
\bibliography{mybibliography}

\iftoggle{arxiv}{
    \clearpage
    \appendix
    \input{appendix}

}{}

\end{document}

%% file: macros.tex
\newcommand{\qedexample}{\hfill$\triangle$}    

\newcommand{\idMat}{I}
\newcommand{\specRad}[1]{\rho(#1)}

\newcommand{\maxnorm}[1]{||#1||_{\infty}}

\newcommand{\allgreater}{\succ}
\newcommand{\allsmaller}{\prec}

\newcommand{\system}[1]{\vec{#1}}
\newcommand{\graphpps}[1]{G_{#1}}
\newcommand{\jac}[1]{#1'}

\renewcommand{\epsilon}{\varepsilon}
\newcommand{\eps}{\varepsilon}

\newcommand{\Nats}{\mathbb{N}}
\newcommand{\Reals}{\mathbb{R}}
\newcommand{\NonNegReals}{\Reals_{\geq 0}}
\newcommand{\NonNegRealsInfty}{\overline{\Reals}_{\geq 0}}
\newcommand{\Rats}{\mathbb{Q}}
\newcommand{\lfp}[1]{\mu #1}

\newcommand{\mc}{\mathcal{M}}

\newcommand{\dist}[1]{\mathcal{D}(#1)}

\newcommand{\prob}{\mathbb{P}}

\newcommand{\pds}{\Delta}
\newcommand{\ppda}{\Delta}
\newcommand{\ppdaStates}{Q}
\newcommand{\ppdaFun}{P}
\newcommand{\ppdaInit}{(\ppdaStates,\abStack,\ppdaFun)}

\newcommand{\trans}[3]{#1 \xrightarrow{#2} #3} 
\newcommand{\abStack}{\Gamma}

\newcommand{\retProb}[3]{[#1#2#3]}

\newcommand{\retProbVar}[3]{\langle#1#2#3\rangle}

\newcommand{\systemRt}[1]{\system f_{#1}^\mathcal{R}}
\newcommand{\systemRp}[1]{\system f_{#1}}
\newcommand{\systemRpClean}[1]{\system{\hat{f}}_{#1}}

\newcommand{\toolname}[1]{\textsc{#1}}
\newcommand{\benchmark}[1]{\textsf{#1}}

%% file: 1_intro.tex
\section{Introduction}

Complex software is likely to contain bugs.
This applies in particular to model checking tools.
This is a serious problem, as the possibility of such bugs compromises the trust one can put in the verification results, rendering the process of formal modeling and analysis less useful.
Ideally, the implementation of a model checker should be formally verified itself~\cite{DBLP:conf/cav/EsparzaLNNSS13}.
However, due to the great complexity of these tools, this is often out of reach in practice.
\emph{Certifying algorithms}~\cite{DBLP:journals/csr/McConnellMNS11} mitigate this problem by providing an \emph{easy-to-check certificate} along with their regular output.
This means that there exists a \emph{verifier} that, given the input problem, the output, and the certificate, constructs a formal proof that the output is indeed correct.
The idea is that the verifier is much simpler than the algorithm, and thus likely to be bug-free or even amenable to formal verification.

\begin{figure}[t]
    \begin{minipage}{0.6\textwidth}
        \begin{align*}
        X &~\rightarrow~ a \mid XYY 			 & x &~=~ \frac{1}{2}( 1 +  xy^2)\\
        Y &~\rightarrow~ b \mid X \mid YY 	 & y &~=~ \frac{1}{3} (1 +  x + y^2)
        \end{align*}
    \end{minipage}
    \hfill
    \begin{minipage}{0.35\textwidth}
        \begin{tikzpicture}[scale=0.8]
        \begin{axis}[
        scale=0.45,
        axis lines = middle,
        xlabel = {$x$},
        ylabel = {$y$},
        xmin=0.4, xmax=1.1,
        ymin=0.4, ymax=1.1]
        
        \addplot [name path = B,
        red,
        domain = 0.1:1.1] {sqrt(2*x - 1) / sqrt(x)};
        
        \addplot [name path = A,
        blue,
        domain = 0.1:1.1,
        samples = 100] {0.3333 * (1 + x + x*x)};
        
        \coordinate (lfp) at (0.6626, 0.7005);
        \node[fill=black,circle,scale=0.4,label={right:\scriptsize $\approx(.66,.7)$}] (lfp) at (lfp) {};
        
        \node[fill=black,circle,scale=0.4,label={left,xshift=-0.5mm:\scriptsize $(1,1)$}]  at (1,1) {};
        
        \addplot [teal!10] fill between [of = A and B, soft clip={domain=0.6626:1.0}];
        
        \end{axis}
        \end{tikzpicture}
    \end{minipage}
    \caption{
        Left: A stochastic context-free grammar (SCFG; e.g.~\cite{DBLP:conf/stacs/EtessamiY05}) and
        the associated positive polynomial system (PPS) which encodes the termination probabilities of each non-terminal, assuming production rules are taken uniformly at random.
        Right: The curves defined by the two equations.
        The least fixpoint (lfp) is $\approx (0.66,0.70)$.
        The thin colored area to the top right of the lfp is the set of inductive,  self-certifying upper bounds on the lfp.
    }
    \label{fig:introExample}
\end{figure}

This paper extends the recent line of research on probabilistic certification~\cite{DBLP:conf/tacas/FunkeJB20,DBLP:phd/dnb/Jantsch22,DBLP:conf/atva/Jantsch0B20,DBLP:conf/tacas/WimmerM20} to \emph{probabilistic pushdown automata}~\cite{DBLP:conf/lics/EsparzaKM04,DBLP:journals/lmcs/KuceraEM06} (pPDA).
pPDA and related models have applications in, amongst others, pattern recognition~\cite{DBLP:journals/prl/SimistiraKC15}, computational biology~\cite{DBLP:journals/nar/KnudsenH03}, and speech recognition~\cite{DBLP:conf/icassp/JurafskyWSSFTM95}.
They are moreover a natural operational model for programs with procedures, recursion, and (discrete) probabilistic constructs such as the ability to flip coins.
With the advent of \emph{probabilistic programming}~\cite{DBLP:journals/corr/abs-1809-10756} as a paradigm for model-based machine learning~\cite{bishop2013model}, such programs have received lots of attention recently.
Moreover, several efficient algorithms such as Hoare's quicksort with randomized pivot selection (e.g.~\cite{DBLP:journals/dam/Karp91}) are readily encoded as probabilistic recursive programs.

A pPDA can be seen as a purely probabilistic variant of a standard pushdown automaton:
Instead of reading an input word, it takes its transitions randomly based on fixed probability distributions over successor states.
Quantities of interest in pPDA include reachability probabilities~\cite{DBLP:conf/lics/EsparzaKM04}, expected runtimes~\cite{DBLP:journals/jcss/BrazdilKKV15}, variances~\cite{DBLP:conf/lics/EsparzaKM05}, satisfaction probabilities of temporal logic formulas~\cite{DBLP:conf/qest/YannakakisE05,DBLP:conf/fossacs/WinklerGK22}, and others (see~\cite{DBLP:journals/fmsd/BrazdilEKK13} for an overview).
pPDA are equivalent to \emph{Recursive Markov Chains}~\cite{DBLP:journals/jacm/EtessamiY09}.
In the past two decades there have been significant research efforts on efficient \emph{approximative} algorithms for pPDA, especially a decomposed variant of \emph{Newton iteration}~\cite{DBLP:conf/stacs/EtessamiY05,DBLP:conf/stoc/KieferLE07,DBLP:conf/stacs/EsparzaKL08,DBLP:journals/jacm/EtessamiY09,DBLP:journals/siamcomp/EsparzaKL10,DBLP:conf/stacs/EsparzaGK10,DBLP:journals/jacm/StewartEY15} which provides guaranteed lower, and occasionally upper~\cite{DBLP:conf/stacs/EsparzaGK10,DBLP:journals/siamcomp/EsparzaKL10} bounds on key quantities.
However, even though implementations might be complex~\cite{premo}, these algorithms are not certifying.

Our technique for certificate generation is a non-trivial extension of \emph{Optimistic Value Iteration}~\cite{ovi} (OVI) to pPDA.
In a nutshell, the idea of OVI is to compute \emph{some} lower bound $\vec l$ on the solution---which can be done using an approximative iterative algorithm---and then \emph{optimistically} \emph{guess} an upper bound $\vec u = \vec l + \vec\eps$ and verify that the guess was correct.
Prior to our paper, OVI had only been considered in Markov Decision Processes (MDP)~\cite{ovi} and Stochastic Games (SG)~\cite{DBLP:conf/atva/AzeemEKSW22}, where it is used to compute bounds on, e.g., maximal reachability probabilities.
The upper bounds computed by OVI have a special property:
They are \emph{self-certifying} (also called \emph{inductive} in our paper):
Given the system and the bounds, one \emph{can check very easily that the bounds are indeed correct}.

However, \emph{pPDA are much more complicated than MDP or SG} for the following reasons:
(i) pPDA may induce \emph{infinite-state} Markov processes due to their unbounded stack;
(ii) the analysis of pPDA requires solving \emph{non-linear equations};
(iii) the complexity of basic decision problems is generally higher than in MDP/SG.
For example, reachability in MDP is characterized as the \emph{least fixpoint} (lfp) of a piece-wise \emph{linear} function that can be computed in PTIME via, e.g., LP solving.
On the other hand, reachability in pPDA requires computing a fixed point of a \emph{positive polynomial} function, leading to a PSPACE complexity bound~\cite{DBLP:conf/lics/EsparzaKM04}.
See \Cref{fig:introExample} for an example.


\paragraph{Contributions.}
Despite the difficulties mentioned above, we show in this paper that
\ul{the general idea of OVI can be extended to pPDA, yielding a practically feasible algorithm with good theoretical properties}.
More concretely:

\begin{contribution}
\emph{
    We present an OVI-style algorithm for computing inductive upper bounds of any desired precision $\eps > 0$ on the lfp of a positive polynomial system.
    Compared to the existing OVI~\cite{ovi}, the key novelty of our algorithm is to compute a certain \emph{direction} $\vec v$ in which to guess, i.e., the guess is $\vec u = \vec l + \eps \vec v$ rather than $\vec u = \vec l + \vec \eps$.
    The direction $\vec v$ is an estimate of a certain \emph{eigenvector}.
    This ensures that we eventually hit an inductive bound, even if the latter lie in a very \enquote{thin strip} as in \Cref{fig:introExample}, and yields a \emph{provably complete} algorithm that is guaranteed to find an inductive bound in finite time (under mild assumptions).
}
\end{contribution}

\begin{contribution}
\emph{
    We implement our algorithm in the software tool \toolname{pray} and compare the new technique to an out-of-the-box approach based on SMT solving, as well as to standard OVI with a simpler guessing heuristic.
}
\end{contribution}

\paragraph{Related Work.}
Certification of pPDA has not yet been addressed explicitly, but some existing technical results go in a similar direction.
For instance, \cite[Prop.\ 8.7]{DBLP:journals/jacm/EtessamiY09} yields certificates for non-termination in SCFG, but they require an SCC decomposition for verification.
\emph{Farkas certificates} for MDP~\cite{DBLP:conf/tacas/FunkeJB20} are more closely related to our idea of certificates.
They require checking a set of \emph{linear} constraints.
A symbolic approach to verify probabilistic recursive programs on the syntax level including inductive proof rules for upper bounds was studied in~\cite{DBLP:conf/lics/OlmedoKKM16}.
A higher-order generalization of pPDA was introduced in~\cite{DBLP:journals/lmcs/KobayashiLG20}, and an algorithm for finding upper bounds inspired by the Finite Element method was proposed.
Applications of PPS beyond the analysis of pPDA include the recent \emph{factor graph grammars}~\cite{DBLP:conf/nips/0001R20} as well as obtaining approximate counting formulas for many classes of trees in the framework of \emph{analytic combinatorics}~\cite{DBLP:books/daglib/0023751}.
Regarding software tools, \toolname{PReMo}~\cite{premo} implements iterative algorithms for lower bounds in Recursive Markov Chains, but it supports neither certificates nor upper bounds.

\paragraph{Paper Outline.}
We review the relevant background information on PPS in \Cref{sec:prelims}.
\Cref{sec:theoryMain} presents our theoretical results on inductive upper bounds in PPS as well as the new  Optimistic Value Iteration algorithm.
In \Cref{sec:ppda} we explain how inductive bounds in PPS are used to certify quantitative properties of pPPA.
The experimental evaluation is in \Cref{sec:experiments}.
We conclude in \Cref{sec:conclusion}.
\iftoggle{arxiv}{}{A full version of this paper is available online~\cite{arxiv}}.

%% file: 2_ppsovi.tex
\section{Preliminaries}
\label{sec:prelims}

\paragraph{Notation for Vectors.}
All vectors in this paper are \emph{column} vectors and are written in boldface, e.g., $\vec u = (u_1, \ldots,  u_n)^T$.
For vectors $\vec u , \vec u'$, we write $\vec u \leq \vec u'$ if $\vec u$ is \emph{component-wise} less than or equal to $\vec u'$.
Moreover, we write $\vec u < \vec u'$ if $\vec u \leq \vec u'$ and $\vec u \neq \vec u'$, and $\vec u \allsmaller \vec u'$ if $\vec u$ is component-wise \emph{strictly} smaller than $\vec u'$.
The zero vector is denoted $\vec 0$.
The \emph{max norm} of a vector $\vec u$ is $\maxnorm{\vec u} = \max_{1\leq i \leq n} |u_i|$.
We say that $\vec u$ is \emph{normalized} if $\maxnorm{\vec u} = 1$.

\paragraph{Positive Polynomial Systems (PPS).}

Let $n \geq 1$ and $\vec x = (x_1 ,\ldots ,x_n)^T$ be a vector of variables.
An $n$-dimensional PPS is an equation system of the form
\[
    x_1 ~=~ f_1(x_1,\ldots,x_n) \quad\ldots\quad x_n ~=~ f_n(x_1,\ldots,x_n)
\]
where for all $1 \leq i \leq n$, the function $f_i$ is a \emph{polynomial with non-negative real coefficients}.
An example PPS is the system $x = \frac{1}{2}(1 + xy^2), y = \frac{1}{3}(1 + x + y^2)$ from \Cref{fig:introExample}.
We also use vector notation for PPS: $\vec x = \system f(\vec x) = (f_1(\vec x), \ldots, f_n(\vec x))^T$.

We write $\NonNegRealsInfty = \NonNegReals \cup \{\infty\}$ for the \emph{extended non-negative reals}.
By convention, for all $a \in \NonNegRealsInfty$, $a \leq \infty$, $a + \infty = \infty + a = \infty$, and $a \cdot \infty = \infty \cdot a$ equals $0$ if $a = 0$ and $\infty$ otherwise.
For $n \geq 1$, the partial order $(\NonNegRealsInfty^n, \leq)$ is a \emph{complete lattice}, i.e., all subsets of $\NonNegRealsInfty^n$ have an infimum and a supremum.
In particular, there exists a least element $\vec 0$ and a greatest element $\vec \infty = (\infty, \ldots, \infty)^T$.
Every PPS induces a \emph{monotone} function $\system{f} \colon \NonNegRealsInfty^n \to \NonNegRealsInfty^n$, i.e., $\vec{u} \leq \vec{v} \implies \system{f}(\vec{u}) \leq \system{f}(\vec{v})$.
By the Knaster-Tarski fixpoint theorem, the set of fixpoints of $\system f$ is also a complete lattice, and thus there exists a \emph{least fixpoint} (lfp) denoted by $\lfp \system f$.

In general, the lfp $\lfp\system f$ is a vector which may contain $\infty$ as an entry.
For instance, this happens in the PPS $x = x+1$.
A PPS $\system{f}$ is called \emph{feasible} if $\lfp\system f \allsmaller \vec \infty$ (or equivalently, $\lfp \system{f} \in \NonNegReals^n$).
The Knaster-Tarski theorem also implies:

\begin{lemma}[Inductive upper bounds]
    \label{thm:induction}
    For all $\vec u \in \NonNegRealsInfty^n$ it holds that
    \begin{align*}
        \system{f}(\vec{u}) \leq \vec u
        \quad\text{implies}\quad
        \lfp \system f \leq \vec u ~.
    \end{align*}
    Such a vector $\vec u$ with $\vec u \allsmaller \vec\infty$ is called \emph{inductive} upper bound.
\end{lemma}
If $\system f$ is feasible, then $\lfp\system f$ is obviously an inductive upper bound.
The problem is that $\lfp\system f$ may be irrational even if $\system f$ has rational coefficients only (see \Cref{ex:ppsBasics} below) and can thus not easily be represented exactly.
In \Cref{sec:theoryMain} we show under which conditions there exist \emph{rational} inductive upper bounds $\vec u \in \Rats_{\geq 0}^n$. 

\vspace{1mm}
\noindent\begin{tikzpicture}
\node[rectangle,rounded corners,draw=black!50,very thick,minimum width=\textwidth, inner sep=2.5mm, align=center] (box) {
    \vspace*{5mm}
    \\
    Given a feasible PPS $\system f$, find a rational \emph{inductive} upper bound $\vec u \geq \lfp\system{f}$.
};
\node[rectangle,fill=white,above=-2mm of box,xshift=-25.5mm,yshift=-1mm] {\textbf{Problem statement of this paper}};
\end{tikzpicture}

A PPS is called \emph{clean} if $\lfp \system f \allgreater \vec{0}$.
Every PPS can be cleaned in linear time by identifying and removing the variables that are assigned 0 in the lfp~\cite{DBLP:journals/jacm/EtessamiY09,DBLP:journals/siamcomp/EsparzaKL10}.

Given a PPS $\system{f}$ and a point $\vec{u} \in \NonNegReals^n$, we define the \emph{Jacobi matrix} of $\system f$ at $\vec u$ as the $n {\times} n$-matrix $\jac{\system f}(\vec u)$ with coefficients
$
    \jac{\system f}(\vec u)_{1 \leq i,j \leq n}
    =
    \frac{\partial}{\partial x_j} f_i(\vec u)
$.

\newcommand{\fExample}{\system f_{ex}}
\newcommand{\fExampleModified}{\system{\tilde{f}}_{\mathit{ex}}}
\begin{example}
    \label{ex:ppsBasics}
    Consider the example PPS $\fExample$ with variables $\vec{x} = (x, y)^T$:
    \begin{align*}
        x ~=~ f_1(x,y) ~=~ y + 0.1 
        && y ~=~ f_2(x,y) ~=~ 0.2x^2 + 0.8 xy + 0.1 ~.
    \end{align*}
    The line and the hyperbola defined by these equations are depicted in \Cref{fig:ppsIllustration} on \Cpageref{fig:ppsIllustration}.
    The fixpoints of $\fExample$ are the intersections of these geometric objects; in this case there are two.
    In particular, $\fExample$ is feasible and its lfp is
    \[
        \lfp \fExample
        ~=~
        \big(\, (27 {-} \sqrt{229}) / 50 \,,\, (22 {-} \sqrt{229}) / 50 \,\big)^T
        ~\approx~
        (0.237 \,,\, 0.137)^T
        ~.
    \]
    Therefore, $\fExample$ is clean as $\lfp \fExample \allgreater \vec 0$.
    The Jacobi matrix of $\fExample$ is 
    \[
        \jac{\fExample}(x,y)
        ~=~
        \begin{pmatrix}
            \frac{\partial}{\partial x}f_1 \,&\, \frac{\partial}{\partial y}f_1 \\
            \frac{\partial}{\partial x}f_2 \,&\, \frac{\partial}{\partial y}f_2
        \end{pmatrix}
        ~=~
        \begin{pmatrix}
            0 \,&\, 1 \\
            0.4x + 0.8y \,&\, 0.8x
        \end{pmatrix}
        ~.
    \]
    Note that the lfp $\lfp\fExample$ contains \emph{irrational} numbers.
    In the above example, these irrational numbers could still be represented using square roots because the fixpoints of $\fExample$ are the zeros of a quadratic polynomial.
    However, there are PPS whose lfp \emph{cannot} be expressed using radicals, i.e., square roots and cubic roots, etc.~\cite{DBLP:conf/stacs/EtessamiY05}.
    This means that in general, there is no easy way to compute the lfp exactly.
    It is thus desirable to provide bounds, which we do in this paper.
    \qedexample
\end{example}

\paragraph{Matrices and Eigenvectors.}
Let $M$ be a real $n{\times} n$-matrix.
We say that $M$ is \emph{non-negative} (in symbols: $M \geq 0$) if it has no negative entries.
$M$ is called \emph{irreducible} if for all $1 \leq i,j \leq n$ there exists $0 \leq k < n$ such that $(M^k)_{i,j} \neq 0$.
It is known that $M$ is irreducible iff the directed graph $G_M = (\{1,\ldots,n\}, E)$ with $(i,j) \in E$ iff $M_{i,j} \neq 0$ is strongly connected.
A \emph{maximal irreducible submatrix} of $M$ is a square submatrix induced by a strongly connected component of $G_M$.
The \emph{period} of a strongly connected $M$ is the length of the shortest cycle in $G_M$.
It is instructive to note that PPS $\vec x = \system f (\vec x)$ are generalizations of linear equation systems of the form $\vec x = M \vec x + \vec c$, with $M \geq 0$ and $\vec c \geq \vec 0$.
Moreover, note that for any PPS $\system f$ it holds that $\jac{\system f}(\vec u) \geq 0$ for all $\vec u \allgreater \vec 0$.

An \emph{eigenvector} of an $n {\times} n$-matrix $M$ with eigenvalue $\lambda \in \mathbb{C}$ is a (complex) vector $\vec v \neq \vec 0$ satisfying $M \vec v = \lambda \vec v$.
There are at most $n$ different eigenvalues.
The \emph{spectral radius} $\specRad{M} \in \NonNegReals$ is the largest absolute value of the eigenvalues of $M$.
The following is a fundamental theorem about non-negative matrices:

\begin{theorem}[{Perron-Frobenius; e.g.~\cite{nonnegmats}}]
    \label{thm:perronFrobenius}
    Let $M \geq 0$ be irreducible.
    \begin{enumerate}[label={(\arabic*)}]
        \item $M$ has a strictly positive eigenvector $\vec{v} \allgreater \vec{0}$ with eigenvalue $\specRad{M}$, the spectral radius of $M$, and all other eigenvectors $\vec{v'} \allgreater \vec{0}$ are scalar multiples of $\vec{v}$.
        \item The eigenvalues of $M$ with absolute value $\specRad{M}$ are exactly the $h$ numbers $\specRad{M}, \xi \specRad{M}, \ldots, \xi^{h-1} \specRad{M}$, where $\xi$ is a primitive $h$th root of unity.
    \end{enumerate}
\end{theorem}
The \emph{unique} eigenvector $\vec v \allgreater \vec 0$ with $\maxnorm{\vec v} = 1$ of an irreducible non-negative matrix $M$ is called the \emph{Perron-Frobenius} eigenvector of $M$.

\paragraph{Strongly Connected Components.}
To each PPS $\system f$ we associate a finite directed graph $\graphpps{\system f} = (\{x_1,\ldots,x_n\}, E)$, which, intuitively speaking, captures the dependency structure among the variables.
Formally, $(x_i, x_j) \in E$ if the polynomial $f_i$ \emph{depends} on $x_j$, i.e., $x_j$ appears in at least one term of $f_i$ with a non-zero coefficient.
This is equivalent to saying that the \emph{partial derivative} $\frac{\partial}{\partial x_j} f_i$ is not the zero polynomial.
We say that $\system{f}$ is \emph{strongly connected} if $\graphpps{\system f}$ is strongly connected, i.e., for each pair $(x_i, x_j)$ of variables, there exists a path from $x_i$ to $x_j$ in $\graphpps{\system f}$.
For instance, $\fExample$ from \Cref{ex:ppsBasics} is strongly connected because the dependency graph has the edges $E = \{(x,y), (y,x), (y,y)\}$.
Strong connectivity of PPS is a generalization of irreducibility of matrices; indeed, a matrix $M$ is irreducible iff the PPS $\vec x = M \vec x$ is strongly connected.
We often use the fact that $\jac{\system f}(\vec u)$ for $\vec u \allgreater \vec 0$ is irreducible iff $\system f$ is strongly connected.

PPS are usually analyzed in a decomposed fashion by considering the sub-systems induced by the \emph{strongly connected components} (SCCs) of $\graphpps{\system f}$ in bottom-up order~\cite{DBLP:conf/stacs/EtessamiY05}.
Here we also follow this approach and therefore focus on strongly connected PPS.
The following was proved in~\cite[Lem.~6.5]{DBLP:journals/jacm/EtessamiY09} and later generalized in~\cite[Thm.~4.1]{DBLP:journals/siamcomp/EsparzaKL10} (also see remark below~\cite[Prop.~5.4]{DBLP:journals/siamcomp/EsparzaKL10} and~\cite[Lem.~8.2]{DBLP:journals/jacm/EtessamiY09}):

\begin{theorem}[{\cite{DBLP:journals/jacm/EtessamiY09,DBLP:journals/siamcomp/EsparzaKL10}}]
    \label{thm:specRadBounded}
    If $\system f$ is feasible, strongly connected and clean, then for all $\vec{u} < \lfp \system f$, we have $\specRad{\jac{\system f}(\vec{u})} < 1$.
    As a consequence, $\specRad{\jac{\system f}(\lfp\system f)} \leq 1$.
\end{theorem}
\Cref{thm:specRadBounded} partitions all PPS $\system f$ which satisfy its precondition into two classes:
Either (1) $\specRad{\jac{\system f}(\lfp\system f)} < 1$, or (2) $\specRad{\jac{\system f}(\lfp\system f)} = 1$.
In the next section we show that $\system f$ admits non-trivial inductive upper bounds iff it is in class (1).

\begin{example}
    Reconsider the PPS $\fExample$ from \Cref{ex:ppsBasics}.
    It can be shown that $\vec v = (1, \lambda_1)^T$ where $\lambda_1 \approx 0.557$ is an eigenvector of $\jac{\fExample}(\lfp \fExample)$ with eigenvalue $\lambda_1$.
    Thus by the Perron-Frobenius Theorem, $\specRad{\jac{\fExample}(\lfp \fExample)} = \lambda_1 < 1$.
    As promised, there exist inductive upper bounds as can be seen in \Cref{fig:ppsIllustration}.
    \qedexample
\end{example}

\section{Finding Inductive Upper Bounds in PPS}
\label{sec:theoryMain}

In this section, we are concerned with the following problem:
Given a feasible, clean, and strongly connected PPS $\system f$, find a vector $\vec 0 \allsmaller \vec u \allsmaller \vec \infty$ such that $\system f(\vec u) \leq \vec u$, i.e., an inductive upper bound on the lfp of $\system f$ (see~\Cref{thm:induction}).

\subsection{Existence of Inductive Upper Bounds}

\begin{figure}[t]
    \centering
    \begin{tikzpicture}[scale=0.6]
    \begin{axis}[
    axis lines = middle,
    xlabel = {$x$},
    ylabel = {$y$},
    xmin=0.05, xmax=0.9,
    ymin=0.05, ymax=0.9,
    legend cell align={left},
    legend style={
        at={(3.75,1.25)},
        font=\scriptsize,
        align=left
    }]
    
    \addplot [name path = B,
    red,
    domain = -0.05:1.0] {x - 0.1};
    \addlegendentry{$x = y + 0.1$}
    
    \addplot [name path = A,
    blue,
    domain = -0.05:0.86,
    samples = 100] {(2*x^2 + 1) / (10 - 8*x)};
    \addlegendentry{$y=0.2x^2 + 0.8 xy + 0.1$}

    \coordinate (lfp) at (0.237, 0.137);
    \coordinate (ev) at (1, 0.557);
    \coordinate (evend) at ($ (lfp) + 0.3*(ev) $);
    \coordinate[xshift=3mm] (cone1) at (evend);
    \coordinate[xshift=3mm,yshift=3mm] (cone2) at (evend);
    \coordinate[yshift=3mm] (cone3) at (evend);
    \node[fill=black,circle,scale=0.3,label={above,xshift=-2mm,yshift=-1mm:$\lfp\fExample$}] (lfp) at (lfp) {};
    
    \fill[black!20] (lfp) -- (cone1) -- (cone2) -- (cone3) -- (lfp);
    \draw[densely dotted] (evend) -- (cone1);
    \draw[densely dotted] (evend) -- (cone3);
    \draw [decorate,decoration={brace,amplitude=2pt}]
    (cone2) -- (cone1) node[right,midway,yshift=2pt] {$\eps$};
    
    \draw[->] (lfp) -- node[above,near end] {$\vec v$} (evend);

    \addplot [name path = C,
    blue,
    dashed,
    domain = -0.05:0.8,
    samples = 100] {(2*x^2 + 1.916) / (10 - 8*x)};
    \addlegendentry{$y=0.2x^2 + 0.8 xy + 0.1916$}
    
    \coordinate (lfp2) at (0.54, 0.44);
    \node[fill=black,circle,scale=0.3,label={above,xshift=-2mm:$\lfp\fExampleModified$}] (lfp2) at (lfp2) {};

    \addplot [teal!10] fill between [of = A and B, soft clip={domain=0.237:0.85}];
    
    \end{axis}
    \end{tikzpicture}
    \caption{
        The PPS $\fExample$ corresponds to the solid red line and the solid blue curve.
        Its inductive upper bounds form the shaded area above the lfp $\lfp\fExample$.
        \Cref{thm:existsInductiveChar}\ref{it:cone} ensures that one can fit the gray \enquote{cone} pointing in direction of the Perron-Frobenius eigenvector $\vec v$ inside the inductive region.
        The PPS $\fExampleModified$ which comprises the dashed curve and the solid line does not have any non-trivial inductive upper bounds.
        Note that the tangent lines at $\lfp\fExampleModified$ are parallel to each other.
    }
    \label{fig:ppsIllustration}
\end{figure}

An important first observation is that \emph{inductive upper bounds other than the exact lfp do not necessarily exist.}
As a simple counter-example consider the 1-dimensional PPS $x = \frac{1}{2}x^2 + \frac{1}{2}$.
If $u$ is an inductive upper bound, then
\[
    \frac{1}{2}u^2 + \frac{1}{2} \leq u
    ~\implies~
    u^2 - 2u + 1 \leq 0
    ~\implies~
    (u-1)^2 \leq 0
    ~\implies~
    u=1 ~,
\]
and thus the only inductive upper bound is the exact lfp $u=1$.
Another example is the PPS $\fExampleModified$ from \Cref{fig:ppsIllustration}.
What these examples have in common is the following property:
Their derivative evaluated \emph{at the lfp} is not invertible.
Indeed, we have $\frac{\partial}{\partial x} (\frac{1}{2}x^2 + \frac{1}{2} - x) = x - 1$, and inserting the lfp $x=1$ yields zero.
The higher dimensional generalization of this property to arbitrary PPS $\system f$ is that the Jacobi matrix of the function $\system f - \vec x$ evaluated at $\lfp \system f$ is singular; note that this is precisely the matrix $\jac{\system f}(\lfp\system f) - \idMat$.
Geometrically, this means that the tangent lines at $\lfp\system f$ are parallel, as can be seen in~\Cref{fig:ppsIllustration} for the example PPS $\fExampleModified$.
It should be intuitively clear from the figure that \emph{inductive upper bounds only exist if the tangent lines are not parallel}.
The next lemma makes this more precise:

\begin{restatable}[Existence of inductive upper bounds]{lemma}{existsInductiveChar}
    \label{thm:existsInductiveChar}
    Let $\system{f}$ be a feasible, clean, and strongly connected PPS.
    Then the following are equivalent:
    \begin{enumerate}[label={(\arabic*)}]
        \item\label{it:nonsing} The matrix $\idMat - \jac{\system f} (\lfp \system{f})$ is non-singular.
        \item\label{it:specrad} The spectral radius of $\jac{\system f} (\lfp \system{f})$ satisfies $\specRad{\jac{\system f} (\lfp \system{f})} < 1$.
        \item\label{it:inductiveNotFp} There exists $\vec 0 \allsmaller \vec u \allsmaller \vec \infty$ s.t. $\system f (\vec u) < \vec{u}$ (i.e. $\vec u$ is inductive but not a fixpoint).
        \item\label{it:cone} The matrix $\jac{\system f} (\lfp \system{f})$ has a unique (normalized) eigenvector $\vec{v} \allgreater \vec{0}$ and there exist numbers $\delta_{max} > 0$ and $\eps > 0$ s.t.
        \[
            \system{f}(\, \lfp\system f + \delta \cdot \vec{\tilde v} \,)
            \quad \allsmaller \quad
            \lfp\system f + \delta \cdot \vec{\tilde v}
        \]
        holds for all $0< \delta \leq \delta_{max}$ and vectors $\vec{\tilde v} \geq \vec{v}$ with $\maxnorm{\vec{v} - \vec{\tilde v} } \leq \eps$.
    \end{enumerate}
\end{restatable}

The proof of \Cref{thm:existsInductiveChar} (see \iftoggle{arxiv}{\Cref{proof:existsInductiveChar}}{\cite{arxiv}}) relies on a linear approximation of $\system f$ via Taylor's familiar theorem as well as \Cref{thm:perronFrobenius,thm:specRadBounded}.
Condition \ref{it:cone} of \Cref{thm:existsInductiveChar} means that there exists a \enquote{truncated cone}
\[
    Cone(\lfp\system f, \vec{v}, \eps, \delta_{max})
    ~=~
    \{\, \lfp{\system f} + \delta \vec{\tilde v} \mid 0 \leq \delta \leq \delta_{max}, \vec{\tilde v} \geq \vec{v}, \maxnorm{\vec{\tilde v} - \vec{v}} \leq \eps \,\}
\] which is entirely contained in the inductive region.
The \enquote{tip} of the cone is located at the lfp $\lfp\system f$ and, the cone points in the direction of the Perron-Frobenius eigenvector $\vec v$, as illustrated in \Cref{fig:ppsIllustration} (assuming $\delta_{max} = 1$ for simplicity).
The length $\delta_{max} > 0$ and the radius $\eps > 0$ of the cone depend on $\specRad{\jac{\system f} (\lfp \system{f})}$, but for us it suffices that they are non-zero.
Note that this cone has non-empty interior and thus contains rational-valued vectors.
The idea of our Optimistic Value Iteration is to \emph{construct a sequence of guesses that eventually hits this cone}.

\subsection{The Optimistic Value Iteration Algorithm}

\newcommand{\tol}{\tau} 
\newcommand{\numIters}{N}
\newcommand{\dampEvScale}{d}
\newcommand{\dampTol}{c}
\newcommand{\currLower}{\vec{l}}
\newcommand{\newLower}{\vec{l'}}
\newcommand{\candidate}{\vec{u}}
\newcommand{\ev}{\vec{v}}
\SetKwInOut{Input}{input}
\SetKwInOut{Output}{output}
\SetKwInOut{Termination}{termination}
\SetKwData{True}{true}
\SetKwData{From}{from}
\SetKwData{To}{to}
\SetKwFunction{Update}{{improveLowerBound}}
\SetKwFunction{Converged}{{converged}}
\SetKwFunction{ApproxEV}{{approxEigenvec}}
\SetKwComment{Comment}{/* }{ */}

The basic idea of Optimistic Value Iteration (OVI) can be applied to monotone functions of the form $\vec \phi \colon \NonNegReals^n \to \NonNegReals^n$ (in \cite{ovi}, $\vec \phi$ is the Bellman operator of an MDP).
Kleene's fixpoint theorem suggests a simple method for approximating the lfp $\lfp \vec\phi$ \emph{from below}:
Simply iterate $\vec\phi$ starting at $\vec 0$, i.e., compute the sequence $\vec l_0 = \vec 0$, $\vec l_1 = \vec\phi(\vec l_0)$, $\vec l_2 = \vec\phi(\vec l_1)$, etc.\footnote{In order for the Kleene seqence to converge to the lfp, i.e., $\lim_{k \to \infty} \vec l_k = \lfp\phi$, it suffices that $\vec\phi$ is $\omega$-continuous. This already implies monotonicity.}
In the context of MDP, this iterative scheme is known as \emph{Value Iteration} (VI).
VI is easy to implement, but it is difficult to decide when to stop the iteration.
In particular, standard stopping criteria such as small absolute difference of consecutive approximations are formally unsound~\cite{DBLP:conf/rp/HaddadM14}.
OVI and other algorithms~\cite{DBLP:conf/cav/Baier0L0W17,DBLP:conf/cav/QuatmannK18} cope with this problem by computing not only a lower but also an \emph{upper} bound on $\lfp\vec\phi$.
In the case of OVI, an upper bound with absolute error $\leq \eps$ is obtained as follows (we omit some details):
\begin{enumerate}[label={(\arabic*)}]
    \item Compute $\vec l_k \leq \lfp\vec\phi$ such that $\maxnorm{\vec l_k - \vec l_{k-1}} \leq \tol$, for some (small) $\tol > 0$.
    \item Guess a candidate upper bound $\vec u = \vec l_k + \vec\eps$.
    \begin{enumerate}
        \item If $\vec\phi(\vec u) \leq \vec u$ holds, i.e., $\vec u$ is inductive, then return $\vec u$.
        \item If not, \emph{refine} $\vec u$ (see~\cite{ovi} for details). If the refined $\vec u$ is still not inductive, then go back to step (1) and try again with $0 < \tol ' < \tol$.
    \end{enumerate}
\end{enumerate}

We present our variant of OVI for PPS as \Cref{alg:ovi}.
The main differences to the above scheme are that (i) we do not insist on Kleene iteration for obtaining the lower bounds $\vec l$, and (ii) we approximate the eigenvector $\vec v$ from condition~\ref{it:cone} of \Cref{thm:existsInductiveChar} and compute the \enquote{more informed} guesses $\vec u = \vec l + \eps \vec v $, for various $\eps$.
\emph{Refining the guesses} as original OVI does \emph{is not necessary} (but see our remarks in \Cref{sec:implDetail} regarding floating point computations).

\begin{algorithm}[t]
    \caption{Optimistic Value Iteration (OVI) for PPS}
    \label{alg:ovi}
    
    \Input{strongly connected clean PPS $\system{f}$; maximum abs.\ error $\eps \in \Rats_{> 0}$}
    \Output{a pair $(\currLower,\candidate)$ of rational vectors s.t. $\currLower \leq \lfp\system{f}$, $\system{f}(\candidate) \leq \candidate$ (hence $\lfp\system f \leq \vec u$), and $\maxnorm{\vec{l} - \vec{u}} \leq \eps$}
    \Termination{guaranteed if $\system{f}$ is feasible and $\idMat-\jac{\system{f}}(\lfp\system{f})$ is non-singular}
    
    $\currLower \gets \vec{0}$ ; $\numIters \gets 0$ \;
    $\tol \gets \eps$ \Comment*[r]{$\tau$ is the current \textit{tolerance}}
    
    \While{\True}{
        $\newLower \gets \Update(\system f,\currLower)$  \Comment*[r]{e.g. Kleene or Newton update} 
        \Comment{guess and verify phase starts here}
        \If{$\maxnorm{\currLower - \currLower'} \leq \tol$}{
            $\ev \gets \ApproxEV(\jac{\system{f}}(\currLower), \tol)$ \Comment*[r]{recall $\ev$ is normalized}
            \For{$k$ \From $0$ \To $\numIters$}{
                $\candidate \gets \currLower + \dampEvScale^{k} \eps \cdot \ev$ \Comment*[r]{optimistic guess, $d \in (0,1)$}
                \If{$\system{f}(\candidate) \leq \candidate$}{
                    \Return $(\currLower, \candidate)$ \Comment*[r]{guess was successful}
                }
            }
            $\numIters \gets \numIters + 1$ \;
            $\tol \gets \dampTol \cdot \tol$ \Comment*[r]{decrease tolerance for next guess, $c \in (0,1)$}
        }
        $\currLower \gets \newLower$ \;
    }
\end{algorithm}

The functions $\Update$ and $\ApproxEV$ used in Algorithm~\ref{alg:ovi} must satisfy the following contracts in order for the algorithm to be correct:
\begin{itemize}
    \item The sequence $\currLower_0 = \vec{0}$, $\currLower_{i+1} = \Update(\system f, \currLower_i)$ for $i \geq 0$, is a monotonically increasing sequence of rational vectors converging to $\lfp \system f$.
    \item $\ApproxEV$ must satisfy the following: Let $M \geq 0$ be an irreducible matrix with (normalized) Perron-Frobenius eigenvector $\vec{v} \allgreater \vec 0$.
    Then for all $\eps > 0$, we require that there exists $\tol > 0$ such that $\maxnorm{\ApproxEV(M, \tol) - \vec v} \leq \eps$.
    In words, $\ApproxEV$ approximates $\vec v$ up to arbitrarily small absolute error if the tolerance $\tol$ is chosen sufficiently small.
    Moreover, $\ApproxEV(M, \tol)$ returns a \emph{rational} vector.
\end{itemize}

In practice, both the Kleene and the Newton~\cite{DBLP:conf/stacs/EtessamiY05,DBLP:journals/jacm/EtessamiY09,DBLP:journals/siamcomp/EsparzaKL10} update operator can be used to implement $\Update$.
We outline a possible implementation of $\ApproxEV$ further below in \Cref{sec:implDetail}.

\begin{example}
    Consider the following PPS $\system f$: $x = \frac{1}{4} x^2 + \frac{1}{8}$, $y = \frac{1}{4}xy + \frac{1}{4}y + \frac{1}{4}$.
    The table illustrates the execution of \Cref{alg:ovi} on $\system f$ with $\eps = 0.1$ and $c = 0.5$:
    \begin{center}
    \setlength{\tabcolsep}{3pt}
    \renewcommand{\arraystretch}{1.0}
    \begin{tabular}{c c c c c c c c c}
        \toprule
        $\#$ & $\numIters$ & $\tol$ & $\currLower$ & $\newLower$ & $\maxnorm{\currLower - \currLower'}$ & $\ev$ & $\candidate$ & $\system{f}(\candidate) \leq \candidate$ \\
        \midrule
        \scriptsize 1 & $0$ & $0.1$ & $(0,0)$ & $(0.4,0.3)$ & $0.4$ & & & \\
        \scriptsize 2 & $0$ & $0.1$ & $(0.4,0.3)$ & $(0.5,0.4)$ & $0.1$ & $(1.0,0.8)$ & $(0.5,0.38)$ & \xmark\\
        \scriptsize 3 & $1$ & $0.05$ & $(0.5,0.4)$ & $(0.55,0.41)$ & $0.05$ & $(1.0, 0.9)$ & $(0.6, 0.49)$ & \cmark \\
        \bottomrule
    \end{tabular}
    \end{center}
    \noindent The algorithm has to improve the lower bound 3 times (corresponding to the 3 lines of the table).
    After the second improvement, the difference between the current lower bound $\currLower_2$ and the new bound $\newLower_2$ does not exceed the current tolerance $\tol_2 = 0.1$ and the algorithm enters the optimistic guessing stage.
    The first guess $\vec u _2$ is not successful.
    The tolerance is then decreased to $\tol_3 = c \cdot \tol_2 = 0.05$ and the lower bound is improved to $\newLower_3$.
    The next guess $\vec u_3$ is inductive.
    \qedexample
\end{example}

\begin{restatable}{theorem}{oviSoundAndComplete}
    \label{thm:oviSoundAndComplete}
    Algorithm~\ref{alg:ovi} is correct: when invoked with a strongly connected clean PPS $\system f$ and $\eps \in \Rats_{> 0}$, then (if it terminates) it outputs a pair $(\currLower,\candidate)$ of rational vectors s.t. $\currLower \leq \lfp\system{f}$, $\system{f}(\candidate) \leq \candidate$, and $\maxnorm{\vec{l} - \vec{u}} \leq \eps$.
    Moreover, if $\system f$ is feasible and $\idMat-\jac{\system{f}}(\lfp\system{f})$ is non-singular, then the algorithm terminates. 
\end{restatable}

The proof of \Cref{thm:oviSoundAndComplete} (see~\iftoggle{arxiv}{\Cref{proof:oviSoundAndComplete}}{\cite{arxiv}}) crucially relies on condition \ref{it:cone} of \Cref{thm:existsInductiveChar} that assures the existence of a \enquote{truncated cone} of inductive bounds centered around the Perron-Frobenius eigenvector of $\jac{\system f}(\lfp \system f)$ (see \Cref{fig:ppsIllustration} for an illustration).
Intuitively, since the lower bounds $\currLower$ computed by the algorithm approach the lfp $\lfp \system f$, the eigenvectors of $\jac{\system f}(\currLower)$ approach those of $\jac{\system f}(\lfp \system f)$.
As a consequence, it is guaranteed that the algorithm eventually finds an eigenvector that intersects the cone.
The inner loop starting on line 7 is needed because the \enquote{length} of the cone is a priori unknown; the purpose of the loop is to scale the eigenvector down so that it is ultimately small enough to fit inside the cone.

\subsection{Considerations for Implementing OVI}
\label{sec:implDetail}

As said earlier, there are at least two options for $\Update$: Kleene or Newton iteration.
We now show that $\ApproxEV$ can be effectively implemented as well.
Further below we comment on floating point arithmetic.

\paragraph{Approximating the Eigenvector.}
A possible implementation of $\ApproxEV$ relies on the \emph{power iteration} method (e.g.~\cite[Thm.~4.1]{saad2011numerical}).
Given a square matrix $M$ and an initial vector $\vec{v}_0$ with $M\vec{v}_0 \neq \vec 0$, power iteration computes the sequence $(\vec v_ i)_{i \geq 0}$ such that for $i > 0$, $\vec{v}_{i} = M \vec{v}_{i-1} / \maxnorm{M \vec{v}_{i-1}}$.

\begin{restatable}{lemma}{powerIteration}
    \label{thm:powerIteration}
    Let $M \geq 0$ be irreducible.
    Then power iteration applied to $M + \idMat$ and any $\vec v _0 > \vec 0$ converges to the Perron-Frobenius eigenvector $\vec{v} \allgreater \vec 0$ of $M$.
\end{restatable}
The convergence rate of power iteration is determined by the ratio $|\lambda_2| / |\lambda_1|$ where $\lambda_1$ and $\lambda_2$ are eigenvalues of largest and second largest absolute value, respectively.
Each time $\ApproxEV$ is called in Algorithm~\ref{alg:ovi}, the result of the previous call to $\ApproxEV$ may be used as initial approximation $\vec v_0$.

\paragraph{Exact vs Floating Point Arithmetic.}
So far we have assumed exact arithmetic for the computations in \Cref{alg:ovi}, but an actual implementation should use floating point arithmetic for efficiency.
However, \emph{this leads to unsound results}.
More specifically, the condition $\system f(\vec u) \leq \vec u$ may hold in floating point arithmetic even though it is actually violated.
As a remedy, we propose to nevertheless run the algorithm with floats, but then verify its output $\vec u$ with exact arbitrary-precision rational arithmetic.
That is, we compute a rational number approximation $\vec u_{\Rats}$ of $\vec u$ and check $\system f(\vec u_{\Rats}) \leq \vec u_{\Rats}$ with exact arithmetic.
If the check fails, we resort to the following refinement scheme which is an instance of the general $k$-induction principle for complete lattices from~\cite{kinduction}:
We iteratively check the conditions
\[
    \system f(\vec u_{\Rats} \sqcap \system f(\vec u_{\Rats})) \leq \vec{u}_{\Rats}
    \,,\quad
    \system f (\vec u_{\Rats} \sqcap \system f(\vec u_{\Rats} \sqcap \system f(\vec u_{\Rats}))) \leq \vec{u}_{\Rats}
    \,,\quad
    \text{and so on,}
\]
where $\sqcap$ denotes pointwise minimum.
If one of the checks is satisfied, then $\lfp \system f \leq \vec u_{\Rats}$~\cite{kinduction}.
This scheme often works well in practice (see \Cref{sec:experiments}).
The original OVI from~\cite{ovi} uses a similar technique to refine its guesses.

%% file: 5_ppda.tex
\section{Certificates for Probabilistic Pushdown Automata}
\label{sec:ppda}

This section shows how the results from \Cref{sec:theoryMain} can be applied to pPDA.
We introduce some additional notation.
For finite sets $A$, $\dist{A}$ denotes the set of \emph{probability distributions} on $A$.
In this section we often denote tuples without parentheses and commata, e.g., we may write $ab$ rather than $(a,b)$.

\begin{definition}[{pPDA~\cite{DBLP:conf/lics/EsparzaKM04}}]
    A probabilistic pushdown automaton (pPDA) is a triple
    $\ppda = \ppdaInit$ where
    $\ppdaStates \neq \emptyset$ is a finite set of states,
    $\abStack \neq \emptyset$ is a finite stack alphabet, and
    $\ppdaFun \colon \ppdaStates \times \abStack \to \dist{\ppdaStates \times \abStack^{\leq 2}}$ is a probabilistic transition function.
\end{definition}
In the following, we often write $\trans{qZ}{p}{r\alpha}$ instead of $\ppdaFun(qZ)(r\alpha) = p$~\cite{DBLP:conf/lics/EsparzaKM04}.
Intuitively, $\trans{qZ}{p}{r\alpha}$ means that if the pPDA is in state $q$ and $Z$ is on top of the stack, then with probability $p$, the pPDA moves to state $r$, pops $Z$ and pushes $\alpha$ on the stack.
More formally, the semantics of a pPDA $\ppda = \ppdaInit$ is a countably infinite Markov chain with state space $\ppdaStates \times\abStack^*$ and transition probability matrix $M$ such that for all $q,r \in \ppdaStates$, $Z \in \abStack$, $\alpha \in \abStack^{\leq 2}$, $\gamma \in \abStack^*$, we have
\begin{align*}
    M(qZ\gamma, r\alpha\gamma) ~=~ \ppdaFun(qZ)(r\alpha) ~,
    \qquad
    M(q \eps, q \eps) ~=~ 1 ~,
\end{align*}
and all other transition probabilities are zero.
This Markov chain, where the initial state is fixed to $qZ$, is denoted $\mc_{\ppda}^{qZ}$ (see \Cref{fig:example_irrat} for an example).
As usual, one can formally define a probability measure $\prob_\ppda^{qZ}$ on the infinite runs of  $\mc_{\ppda}^{qZ}$ via the standard cylinder construction (e.g.,~\cite[Sec.\ 10]{DBLP:books/daglib/0020348}).

Consider a triple $qZr \in Q {\times} \Gamma {\times} Q$.
We define the \emph{return probability}\footnote{
    See~\cite{DBLP:conf/fossacs/WinklerGK22} for an explanation of this terminology.
} $\retProb{q}{Z}{r}$ as the probability of reaching $r\varepsilon$ in the Markov chain $\mc_{\pds}^{qZ}$, i.e., $\retProb{q}{Z}{r}  =  \prob_\ppda^{qZ} (\lozenge \{r\eps\})$, where $\lozenge \{r\eps\}$ is the set of infinite runs of $\mc_{\ppda}^{qZ}$ that eventually hit state $r\eps$.

\newcommand{\ppdaExample}{\ppda_{ex}}
\begin{figure}[t]
    \centering
    \begin{minipage}{0.23\textwidth}
        \begin{tikzpicture}[on grid, node distance=12mm,thick, every state/.style={inner sep=3pt, minimum size=1pt}]
        \node[rectangle,rounded corners,draw=black,minimum size=4mm] (p0) {$q$};
        \node[rectangle,rounded corners,draw=black, below=of p0,minimum size=4mm] (p1) {$r$};
        
        \draw[->] (p0) edge[loop right] node[auto]{$(\nicefrac{1}{2}, Z, \eps)$} (p0);
        \draw[->] (p0) edge[loop above] node[right=1.5mm] {$(\nicefrac{1}{4}, Z, ZZ)$} (p0);
        \draw[->] (p0) edge node[auto]{$(\nicefrac{1}{4}, Z, \eps)$} (p1);
        \draw[->] (p1) edge[loop right] node[auto]{$(1, Z, \eps)$} (p1);
        \end{tikzpicture}
    \end{minipage}
    \hspace{5mm}
    \begin{minipage}{0.55\textwidth}
        \vspace{3mm}
        \begin{tikzpicture}[initial text=, initial where=above, on grid, node distance=12mm and 20mm,thick]
        \node[rectangle,rounded corners,draw=black,minimum size=4mm] (p0) {$q\eps$};
        \node[initial,right=of p0,rectangle,rounded corners,draw=black,minimum size=4mm] (p0Z) {$qZ$};
        \node[right=of p0Z,rectangle,rounded corners,draw=black,minimum size=4mm] (p0ZZ) {$qZZ$};
        \node[right=of p0ZZ,rectangle,rounded corners,draw=black,minimum size=4mm] (p0ZZZ) {$\cdots$};
        \node[below=of p0,rectangle,rounded corners,draw=black,minimum size=4mm] (p1) {$r \varepsilon$};
        \node[below=of p0Z,rectangle,rounded corners,draw=black,minimum size=4mm] (p1Z) {$rZ$};
        \node[below=of p0ZZ,rectangle,rounded corners,draw=black,minimum size=4mm] (p1ZZ) {$rZZ$};
        \node[below=of p0ZZZ,rectangle,rounded corners,draw=black,minimum size=4mm] (p1ZZZ) {$\cdots$};
        
        \draw[->] (p0Z) edge[bend left] node[auto] {$\nicefrac 1 4$} (p0ZZ);
        \draw[->] (p0ZZ) edge[bend left] node[auto] {$\nicefrac 1 4$} (p0ZZZ);
        \draw[->] (p0Z) edge[] node[auto] {$\nicefrac 1 2$} (p0);
        \draw[->] (p0ZZ) edge[bend left] node[above] {$\nicefrac 1 2$} (p0Z);
        \draw[->] (p0ZZZ) edge[bend left] node[above] {$\nicefrac 1 2$} (p0ZZ);
        \draw[->] (p0Z) edge node[auto] {$\nicefrac 1 4$} (p1);
        \draw[->] (p0ZZ) edge node[auto] {$\nicefrac 1 4$} (p1Z);
        \draw[->] (p0ZZZ) edge node[auto] {$\nicefrac 1 4$} (p1ZZ);
        \draw[->] (p1Z) edge node[auto] {$1$} (p1);
        \draw[->] (p1ZZ) edge node[auto] {$1$} (p1Z);
        \draw[->] (p1ZZZ) edge node[auto] {$1$} (p1ZZ);
        \draw[->] (p0) edge[loop above] node[auto]{$1$} (p0);
        \draw[->] (p1) edge[loop above] node[auto]{$1$} (p1);
        \end{tikzpicture}
    \end{minipage}
    \medskip
    \begin{align*}
        \retProbVar{q}{Z}{q} ~=~ & \nicefrac{1}{4}  \big(\retProbVar{q}{Z}{q}\retProbVar{q}{Z}{q} + \retProbVar{q}{Z}{r}\retProbVar{r}{Z}{q}\big) + \nicefrac{1}{2} & \retProbVar{r}{Z}{q} ~=~ 0 \\
        \retProbVar{q}{Z}{r} ~=~ & \nicefrac{1}{4}  \big(\retProbVar{q}{Z}{q}\retProbVar{q}{Z}{r} + \retProbVar{q}{Z}{r}\retProbVar{r}{Z}{r} \big) + \nicefrac{1}{4} & \retProbVar{r}{Z}{r} ~=~ 1
    \end{align*} 
    \caption{
        Top left: The pPDA $\ppdaExample = (\{q, r\}, \{Z\}, \ppdaFun)$ where $\ppdaFun$ comprises the transitions
        $\trans{qZ}{1/4}{qZZ},\,
        \trans{qZ}{1/2}{q\eps},\,
        \trans{qZ}{1/4}{r\eps},\,
        \trans{rZ}{1}{r\eps}$.
        Top right: A fragment of the \emph{infinite} underlying Markov chain $\mc_{\ppda}^{qZ}$, assuming initial configuration $qZ$.
        Bottom: The associated equation system from \Cref{thm:retProbEqs}.
    }
    \label{fig:example_irrat}
\end{figure}

\begin{theorem}[{The PPS of return probabilities~\cite{DBLP:conf/lics/EsparzaKM04}\footnote{We refer to \cite[Sec.~3]{DBLP:journals/lmcs/KuceraEM06} for an intuitive explanation of the equations in $\systemRp{\ppda}$.}}]
    \label{thm:retProbEqs}
    Let $\ppda = \ppdaInit$ be a pPDA and $(\retProbVar{q}{Z}{r})_{qZr \,\in\, \ppdaStates\times\abStack\times\ppdaStates}$ be variables.
    For each $\retProbVar{q}{Z}{r}$, define
    \[
        \retProbVar{q}{Z}{r}
        \quad=\quad
        \sum_{\trans{qZ}{p}{sYX}} p \cdot \sum_{t \in Q} \retProbVar{s}{Y}{t} \cdot \retProbVar{t}{X}{r}
        ~+~
        \sum_{\trans{qZ}{p}{sY}} p \cdot \retProbVar{s}{Y}{r}
        ~+~
        \sum_{\trans{qZ}{p}{r\eps}} p
    \]
    and call the resulting PPS $\systemRp{\ppda}$.
    Then $\lfp\systemRp{\ppda} = (\retProb{q}{Z}{r})_{qZr \,\in\, \ppdaStates\times\abStack\times\ppdaStates}$.
\end{theorem}

\begin{example}
    \label{ex:irrat_probs}
    \Cref{fig:example_irrat} shows a pPDA $\ppdaExample$ and the associated PPS $\systemRp{\ppdaExample}$.
    The least non-negative solution is $\retProbVar{q}{Z}{q} = 2 - \sqrt{2} \approx 0.586$ and $\retProbVar{q}{Z}{r} = \sqrt{2} - 1 \approx 0.414$ (and, of course, $\retProbVar{r}{Z}{q} = 0$, $\retProbVar{r}{Z}{r} = 1$).
    Thus by \Cref{thm:retProbEqs}, the return probabilities are $\retProb{q}{Z}{q} = 2 - \sqrt{2}$ and $\retProb{q}{Z}{r} = \sqrt{2} - 1$.
    \qedexample
\end{example}

The PPS $\systemRp{\ppda}$ is always feasible (because $\lfp\systemRp{\ppda} \leq \vec 1$).
$\systemRp{\ppda}$ is neither necessarily strongly connected nor clean.
Let $\systemRpClean{\ppda}$ denote the cleaned up version of $\systemRp{\ppda}$.

\begin{restatable}[Basic Certificates for pPDA]{proposition}{basicCertificate}
    \label{thm:basicCertificate}
    A \emph{basic certificate} for $\ppda = \ppdaInit$ is a \emph{rational} inductive upper bound $\vec u \in \Rats_{\geq 0}^{\ppdaStates \times \abStack \times \ppdaStates}$ on the lfp of the return probabilities system $\systemRp{\ppda}$ (see~Thm.~\ref{thm:retProbEqs}).
    They have the following properties:
    \begin{itemize}
        \item (Existence) $\forall \eps > 0$ there exists a basic certificate $\vec u$ with $\maxnorm{\lfp \systemRp{\ppda} - \vec u} \leq \eps$ if all maximal irreducible submatrices $M$ of $\jac{\systemRpClean{\ppda}}(\lfp \systemRpClean{\ppda})$ satisfy $\specRad{M} < 1$.
        \item (Complexity) Let $\beta$ be the maximum number of bits used to encode any of the numerators and denominators of the fractions occurring in $\vec u \in \Rats_{\geq 0}^{\ppdaStates \times \abStack \times \ppdaStates}$.
        Then checking $\systemRp{\ppda}(\vec u) \leq \vec u$, i.e., whether $\vec u$ is basic certificate for $\ppda$, can be done in time polynomial in $\beta$ and the size of $\ppda$.
    \end{itemize}
\end{restatable}
Existence of basic certificates follows from \Cref{thm:existsInductiveChar}  applied to each SCC of the cleaned-up version of $\systemRp{\ppda}$ individually.
However, note that in order to merely \emph{check} the certificate, i.e., verify the inequality $\system f(\vec u) \leq \vec u$, neither do SCCs need to be computed nor does the system has to be cleaned up.

\begin{example}
    \label{ex:basicCertificate}
    Reconsider the example pPDA and its associated (non-strongly connected) system of return probabilities from \Cref{fig:example_irrat}.
    We verify that $\vec u_{qZq} = 3/5$ and $\vec u_{qZr} = 1/2$ (as well as $\vec u_{rZq} = 0, \vec u_{rZr} = 1$) is a basic certificate:
    \begin{align*}
        \frac{1}{4}  \left( \frac{3}{5} \cdot \frac{3}{5}  + \frac{1}{2} \cdot 0 \right) + \frac{1}{2} = \frac{59}{100} ~\overset{\checkmark}{\leq}~ \frac{3}{5}
        \quad,\quad
        \frac{1}{4}  \left(\frac{3}{5} \cdot \frac 1 2 + \frac 1 2 \cdot 1 \right) + \frac{1}{4} = \frac{45}{100} ~\overset{\checkmark}{\leq}~  \frac{1}{2} ~.
    \end{align*}
    Note that $\retProb{q}{Z}{q} \approx 0.586 \leq 3/5 = 0.6$ and $\retProb{q}{Z}{r} \approx 0.414 \leq 1/2 = 0.5$.
    \qedexample
\end{example}

In the following we outline how a variety of key quantities associated with a pPDA can be verified using basic certificates.

\paragraph{Upper Bounds on Temporal Properties.}
We may use basic certificates to verify that a bad state $r_{bad}$ is reached with low probability, e.g., at most $p = 0.01$.
To this end, we remove the outgoing transitions of $r_{bad}$ and add the transitions $\trans{r_{bad}Z}{1}{r_{bad}\eps}$ for all $Z \in \abStack$.
Clearly, $r_{bad}$ is reached with probability at most $p$ from initial configuration $qZ$ iff $\retProb{q}{Z}{r_{bad}} \leq p$.
The results of~\cite{DBLP:conf/lics/EsparzaKM04} imply that this idea can be generalized to \emph{until}-properties of the form $\mathcal{C}_1 \,\mathcal{U}\, \mathcal{C}_2$, where $\mathcal{C}_1$ and $\mathcal{C}_2$ are \emph{regular} sets of configurations.

\paragraph{Certificates for the Output Distribution.}
Once a pPDA reaches the empty stack, we say that it has \emph{terminated}.
When modeling procedural programs, this corresponds to returning from a program's main procedure.
Assuming initial configuration $qZ$, the probability sub-distribution over the possible return values is then given by the return probabilities $\{\retProb{q}{Z}{r} \mid r \in \ppdaStates\}$.
Missing probability mass models the probability of non-termination.
Therefore, a basic certificate may be used to prove a point-wise upper bound on the output distribution as well as \underline{non} \emph{almost-sure termination} (AST).
If a pPDA $\ppda$ is known to be AST, then we can also certify a lower bound on the output distribution:
Suppose that $\vec u$ is a basic certificate for $\ppda$ and assume that $\ppda$ is AST from initial configuration $qZ$.
Define $\eps = \sum_{r \in \ppdaStates} \vec u_{qZr} - 1$.
Then for all $r \in \ppdaStates$, we have
$\vec u_{qZr} - \eps ~\leq~ \retProb{q}{Z}{r} ~\leq~ \vec{u}_{qZr}$.

\begin{example}
    The pPDA $\ppdaExample$ from \Cref{fig:example_irrat} is AST from initial configuration $qZ$, as the transition $\trans{qZ}{1/4}{r\eps}$ is eventually taken with probability 1, and the stack is emptied certainly once $r$ is reached.
    Using the basic certificate from \Cref{ex:basicCertificate} we can thus (correctly) certify that
    $0.5 \leq \retProb{q}{Z}{q} \leq 0.6$ and $0.4 \leq \retProb{q}{Z}{r} \leq 0.5$.
\end{example}

\newcommand{\exRew}[3]{E_{#1#2#3}}
\newcommand{\exRewTermPaths}[2]{E(#1#2)}
\newcommand{\exRewVar}[3]{\langle E_{#1#2#3}\rangle}
\newcommand{\rewFun}{R}
\newcommand{\systemRew}[2]{\system{f}_{#1,#2}}
\newcommand{\rvRew}{V_{\rewFun}}
\newcommand{\firstHit}[1]{\mathit{firstHit(#1)}}

\paragraph{Certificates for Expected Rewards.}
pPDA may also be equipped with a reward function $\ppdaStates \to \NonNegReals$.
It was shown in~\cite{DBLP:conf/lics/EsparzaKM05} that the expected reward
accumulated during the run of a pPDA is the solution of a linear equation system whose coefficients depends on the numbers $\retProb{q}{Z}{r}$.
Given a basic certificate $\vec u$, we obtain an equation system whose solution is an over-approximation of the true expected reward (see~\iftoggle{arxiv}{\Cref{app:pPDA_certification_details}}{\cite{arxiv}}).
We may extend the basic certificate $\vec u$ by the solution of this linear system to make verification straightforward.
Note that a program's \emph{expected runtime}~\cite{DBLP:journals/jcss/BrazdilKKV15,DBLP:conf/lics/OlmedoKKM16} is a special case of total expected reward.

%% file: 5_implementation.tex
\section{Implementation and Experiments}
\label{sec:experiments}

\paragraph{Our Tool: \toolname{pray}.}
We implemented our algorithm in the prototypical Java-tool \textsc{pray} (Probabilistic Recursion AnalYzer)~\cite{zenodo}.
It supports two input formats:
(i) Recursive probabilistic programs in a Java-like syntax (e.g.~\Cref{fig:and-or}); these programs are \emph{automatically} translated to pPDA.
(ii) Explicit PPS in the same syntax used by the tool \textsc{PReMo}~\cite{premo}.
The output of \toolname{pray} is a rational \emph{inductive} upper bound on the lfp of the return probability PPS of the input program's pPDA model (a basic certificate), or on the lfp of the explicitly given PPS.
The absolute precision $\eps$ is configurable.
The implementation works as follows:
\begin{enumerate}[label={(\arabic*)}]
    \item It parses the input and, if the latter is a program, constructs a pPDA model and the associated PPS of return probabilities.
    \item It computes an SCC decomposition of the PPS under consideration using standard algorithms implemented in the \toolname{jGraphT} library~\cite{jgrapht}.
    \item\label{implStage:main} It applies \Cref{alg:ovi} to the individual SCC in reverse topological order using floating point arithmetic.
    \Cref{alg:ovi} is instantiated with Kleene iteration\footnote{In fact, we use the slightly optimized Gauss-Seidel iteration (see~\cite[Sec.~5.2]{DBLP:phd/ethos/Wojtczak09}) which provides a good trade-off between ease of implementation and efficiency~\cite{DBLP:phd/ethos/Wojtczak09}.}, the power iteration for approximating eigenvectors as outlined in \Cref{sec:implDetail}, and constants $c=0.1$, $d=0.5$.
    We allow $\leq 10$ guesses per SCC.
    \item If stage \ref{implStage:main} is successful, the tool verifies the resulting floating point certificate using exact rational number arithmetic as described in \Cref{sec:implDetail}.
\end{enumerate}

\paragraph{Baselines.}
To the best of our knowledge, no alternative techniques for finding \emph{inductive} upper bounds in PPS have been described explicitly in the literature.
However, there is an (almost) out-of-the-box approach using an SMT solver:
Given a PPS $\vec x = \system f(\vec x)$, compute some lower bound $\currLower \leq \lfp\system f$ using an iterative technique.
Then query the SMT solver for a model (variable assignment) of the quantifier-free first-order logic formula
$
    \varphi_{\system f}(\vec x) = \bigwedge_{i=1}^n f_i(\vec x) \leq x_i  \wedge \currLower_i \leq x_i \leq \currLower_i + \eps
$
in the (decidable) theory of polynomial real arithmetic with inequality (aka QF\_NRA in the SMT community).
If such a model $\vec u$ exists, then clearly $\lfp\system f \leq \vec u$ and $\maxnorm{\vec l - \vec u} \leq \eps$.
If no model exists, then improve $\vec l$ and try again.
We have implemented this approach using the state-of-the-art SMT solvers \toolname{cvc5}~\cite{DBLP:conf/tacas/BarbosaBBKLMMMN22} and \toolname{z3}~\cite{DBLP:conf/tacas/MouraB08}, the winners of the 2022 SMT-COMP in the category QF\_NRA\footnote{\url{https://smt-comp.github.io/2022/results}}.

As yet another baseline, we have also implemented a variant of OVI for PPS which is closer to the original MDP algorithm from~\cite{ovi}.
In this variant, called \enquote{standard OVI} from now on, we compute the candidate $\vec u$ based on the \emph{relative} update rule $\vec u = (1 + \eps)\vec l$, where $\vec l$ is the current lower bound~\cite{ovi}.

\paragraph{Research Questions.}
We aim to shed some light on the following questions:
(A) How well does our algorithm scale?
(B) Is the algorithm suitable for PPS with different characteristics, e.g., dense or sparse?
(C) Is the requirement $\specRad{\jac{\system f(\lfp\system f)}} < 1$ restrictive in practice?
(D) How does our OVI compare to the baselines?

\begin{figure}[t]
    \centering
    \begin{adjustbox}{max height=12mm}
        \begin{minipage}{0.49\textwidth}
            \begin{verbatim}
bool and() {
  prob {
    1//2: return
      (1//2: true | 1//2: false);
    1//2: {
      if(!or()) return false;
      else return or(); } } }
            \end{verbatim}
        \end{minipage}
        \begin{minipage}{0.49\textwidth}
            \begin{verbatim}
bool or() {
  prob {
    1//2: return
      (1//2: true | 1//2: false);
    1//2: {
      if(and()) return true;
      else return and(); } } }
            \end{verbatim}
        \end{minipage}	
    \end{adjustbox}
    
    \caption{
        Program evaluating a random and-or tree~\cite{DBLP:journals/jcss/BrazdilKKV15}.
        The \texttt{prob}-blocks execute the contained statements with the respective probabilities (syntax inspired by Java's \texttt{switch}).
        Our tool automatically translates this program to a pPDA and computes a basic certificate (\Cref{thm:basicCertificate}) witnessing that calling \texttt{and()} returns \texttt{true} and \texttt{false} with probability $\leq 382/657 \approx 0.58$ and $391/933 \approx 0.42$, resp.
    }
    \label{fig:and-or}
\end{figure}

\paragraph{Benchmarks.}
To answer the above questions we run our implementation on two sets of benchmarks (\Cref{table:programs} and \Cref{table:pps}, respectively).
The first set consists of various example programs from the literature as well as a few new programs, which are automatically translated to pPDA.
This translation is standard and usually takes not more than a few seconds.
The programs
\benchmark{golden},
\benchmark{and-or} (see~\Cref{fig:and-or}),
\benchmark{virus},
\benchmark{gen-fun} are adapted from
\cite{%
    DBLP:conf/lics/OlmedoKKM16,%
    DBLP:journals/jcss/BrazdilKKV15,%
    DBLP:conf/fossacs/WinklerGK22%
} and~\cite[Program~5.6]{DBLP:journals/corr/abs-1809-10756}, respectively.
The source code of all considered programs is in~\iftoggle{arxiv}{\Cref{app:programs}}{\cite{arxiv}}.
We have selected only programs with possibly unbounded recursion depth which induce \emph{infinite} Markov chains.
The second benchmark set comprises explicit PPS from~\cite{premo}.
The instances \benchmark{brown}, \benchmark{lemonde}, \benchmark{negra}, \benchmark{swbd}, \benchmark{tiger}, \benchmark{tuebadz}, and \benchmark{wsj} all encode SCFG from the area of language processing (see~\cite{premo} for details).
\benchmark{random} is the return probability system of a randomly generated pPDA.

\begin{table}[t]
    \caption{
        Experiments with PPS obtained from recursive probabilistic programs.
        Columns \emph{vars} and \emph{terms} display the number of variables and terms in the PPS.
        Columns \emph{sccs} and \emph{scc$_{max}$} indicate the number of non-trivial SCC and the size of the largest SCC.
        $G$ is total number of guesses made by OVI (at least one guess per SCC).
        $t_{tot}$ is the total runtime excluding the time for model construction.
        $t_{\Rats}$ is the percentage of $t_{tot}$ spent on exact rational arithmetic.
        $D$ is the average number of decimal digits of the rational numbers in the certificate.
        The timeout (TO) was set to 10 minutes.
        Time is in ms.
        The absolute precision is $\eps = 10^{-3}$.
    }
    \label{table:programs}
    \centering
    \setlength{\tabcolsep}{3pt}
    \begin{adjustbox}{max width=\textwidth}
        \rowcolors{1}{}{lightgray}
        \begin{tabular}{    l        r      r      r      r      r      r      r   |  c      r      r      r      r    | c      r      r    | c      r      r      r }
            \toprule
            benchmark             & $|Q|$    & $|P|$    & $|\Gamma|$ & vars     & terms    & sccs     & scc$_{max}$ & cert     & $G$      & $D$      & $t_{\mathbb{Q}}$ & $t_{tot}$ & $cert_{\toolname{z3}}$ & $D_{\toolname{z3}}$ & $t_{\toolname{z3}}$ & $cert_{std}$ & $G_{std}$ & $D_{std}$ & $t_{std}$\\ 
            \midrule
            \benchmark{rw-0.499}  & 18       & 29       & 5        & 38       & 45       & 1        & 12       & \cmark   & 5        & 5        & 17\%     & 163      & \cmark   & 2        & 11       & \cmark   & 4        & 5        & 59      \\ 
            \benchmark{rw-0.500}  & 18       & 29       & 5        & 38       & 45       & 1        & 12       & \xmark   & 10       & -        & -        & 7327     & \cmark   & 2        & 10       & \xmark   & 10       & -        & 8083    \\ 
            \benchmark{rw-0.501}  & 18       & 29       & 5        & 38       & 45       & 1        & 12       & \cmark   & 5        & 4        & 6\%      & 36       & \cmark   & 13       & 12       & \cmark   & 4        & 5        & 23      \\ 
      \benchmark{geom-offspring}  & 24       & 40       & 5        & 52       & 80       & 4        & 24       & \cmark   & 8        & 6        & 13\%     & 15       & \cmark   & 9        & 16       & \cmark   & 8        & 6        & 14      \\ 
            \benchmark{golden}    & 27       & 49       & 6        & 81       & 94       & 1        & 36       & \cmark   & 1        & 5        & 30\%     & 10       & \cmark   & 7        & 14       & \cmark   & 2        & 4        & 12      \\ 
            \benchmark{and-or}    & 50       & 90       & 7        & 149      & 182      & 1        & 48       & \cmark   & 2        & 4        & 26\%     & 19       & \cmark   & 12       & 15260    & \cmark   & 2        & 4        & 19      \\ 
            \benchmark{gen-fun}   & 85       & 219      & 7        & 202      & 327      & 1        & 16       & \cmark   & 2        & 3        & 32\%     & 22       & \cmark   & 15       & 141      & \cmark   & 2        & 3        & 21      \\ 
            \benchmark{virus}     & 68       & 149      & 27       & 341      & 551      & 1        & 220      & \cmark   & 1        & 5        & 38\%     & 40       & \cmark   & 7        & 139      & \cmark   & 1        & 6        & 59      \\ 
            \benchmark{escape10}  & 109      & 174      & 23       & 220      & 263      & 1        & 122      & \cmark   & 1        & 4        & 5\%      & 56       & \cmark   & 7        & 48       & \cmark   & 1        & 8        & 71      \\ 
            \benchmark{escape25}  & 258      & 413      & 53       & 518      & 621      & 1        & 300      & \cmark   & 1        & 5        & 17\%     & 245      & \cmark   & 7        & 15958    & \cmark   & 1        & 9        & 172     \\ 
            \benchmark{escape50}  & 508      & 813      & 103      & 1018     & 1221     & 1        & 600      & \cmark   & 1        & 7        & 23\%     & 653      & \cmark   & 7        & 410      & \xmark   & 1        & -        & 400     \\ 
            \benchmark{escape75}  & 760      & 1215     & 153      & 1522     & 1825     & 1        & 904      & \cmark   & 2        & 9        & 10\%     & 3803     & \xmark   & -        & TO       & \xmark   & 1        & -        & 635     \\ 
           \benchmark{escape100}  & 1009     & 1614     & 203      & 2020     & 2423     & 1        & 1202     & \xmark   & 5        & -        & -        & 29027    & \cmark   & 6        & 939      & \xmark   & 1        & -        & 901     \\ 
           \benchmark{escape200}  & 2008     & 3213     & 403      & 4018     & 4821     & 1        & 2400     & \xmark   & 6        & -        & -        & 83781    & \xmark   & -        & TO       & \xmark   & 1        & -        & 2206    \\ 
         \benchmark{sequential5}  & 230      & 490      & 39       & 1017     & 1200     & 10       & 12       & \cmark   & 15       & 4        & 26\%     & 103      & \cmark   & 8        & 1074     & \cmark   & 15       & 5        & 204     \\ 
         \benchmark{sequential7}  & 572      & 1354     & 137      & 3349     & 3856     & 14       & 12       & \cmark   & 21       & 5        & 27\%     & 1049     & \cmark   & 8        & 12822    & \cmark   & 20       & 5        & 1042    \\ 
        \benchmark{sequential10}  & 3341     & 8666     & 1036     & 26367    & 29616    & 20       & 12       & \cmark   & 30       & 5        & 2\%      & 100613   & \cmark   & 8        & 453718   & \cmark   & 30       & 6        & 101554  \\ 
            \benchmark{mod5}      & 44       & 103      & 10       & 296      & 425      & 1        & 86       & \cmark   & 1        & 5        & 39\%     & 28       & \cmark   & 9        & 34150    & \xmark   & 2        & -        & 178     \\ 
            \benchmark{mod7}      & 64       & 159      & 14       & 680      & 1017     & 1        & 222      & \cmark   & 1        & 6        & 69\%     & 172      & \cmark   & 7        & 443      & \xmark   & 2        & -        & 624     \\ 
            \benchmark{mod10}     & 95       & 244      & 20       & 1574     & 2403     & 1        & 557      & \xmark   & 1        & -        & -        & 675      & \cmark   & 7        & 1245     & \xmark   & 2        & -        & 882     \\ 
            \bottomrule
        \end{tabular}
    \end{adjustbox}
\end{table}

\begin{table}[h]
    \caption{
        Experiments with explicitly given PPS (setup as in \Cref{table:programs}).
    }
    \label{table:pps}
    \centering
    \setlength{\tabcolsep}{4pt}
    \begin{adjustbox}{max width=\textwidth}
        \rowcolors{1}{}{lightgray}
        \begin{tabular}{    l        r      r      r      r   |   c      r      r      r      r   |  c      r      r   |  c      r      r      r }
            \toprule
            benchmark             & vars     & terms    & sccs     & scc$_{max}$ & cert     & $G$      & $D$      & $t_{\mathbb{Q}}$ & $t_{tot}$ & $cert_{\toolname{z3}}$ & $D_{\toolname{z3}}$ & $t_{\toolname{z3}}$ & $cert_{std}$ & $G_{std}$ & $D_{std}$ & $t_{std}$\\ 
            \midrule
            \benchmark{brown}     & 37       & 22866    & 1        & 22         & \cmark   & 2        & 6        & 74\%     & 3212     & \xmark   & -        & TO       & \cmark   & 2        & 8        & 9065    \\ 
            \benchmark{lemonde}   & 121      & 32885    & 1        & 48         & \cmark   & 2        & 5        & 97\%     & 40738    & \xmark   & -        & TO       & \cmark   & 2        & 5        & 38107   \\ 
            \benchmark{negra}     & 256      & 29297    & 1        & 149        & \cmark   & 2        & 7        & 89\%     & 10174    & \cmark   & 1        & 37248    & \cmark   & 1        & 7        & 8873    \\ 
            \benchmark{swbd}      & 309      & 47578    & 1        & 243        & \cmark   & 1        & 7        & 93\%     & 18989    & \xmark   & -        & TO       & \cmark   & 1        & 8        & 67314   \\ 
            \benchmark{tiger}     & 318      & 52184    & 1        & 214        & \cmark   & 2        & 8        & 98\%     & 94490    & \cmark   & 1        & 17454    & \cmark   & 1        & 8        & 90801   \\ 
            \benchmark{tuebadz}   & 196      & 8932     & 2        & 168        & \cmark   & 4        & 9        & 85\%     & 2666     & \cmark   & 1        & 15323    & \cmark   & 3        & 9        & 2700    \\ 
            \benchmark{wsj}       & 240      & 31170    & 1        & 194        & \cmark   & 2        & 9        & 96\%     & 30275    & \xmark   & -        & TO       & \cmark   & 2        & 9        & 29038   \\ 
            \benchmark{random}    & 10000    & 20129    & 1        & 8072       & \cmark   & 3        & 7        & 5\%      & 17585    & \xmark   & -        & TO       & \cmark   & 4        & 8        & 16357   \\
            \bottomrule
        \end{tabular}
    \end{adjustbox}
\end{table}

\paragraph{Summary of Results.}
We ran the experiments on a standard notebook.
The approach based on \toolname{cvc5} turns out to be not competitive (see \iftoggle{arxiv}{\Cref{app:z3vscvc5}}{\cite{arxiv}}).
We thus focus on \toolname{z3} in the following.
Both \toolname{pray} and the \toolname{z3} approach handle most of the programs from \Cref{table:programs} within a 10 minute time limit.
The considered programs induce sparse PPS with 38 - 26,367 variables, and most of them have just a single SCC.
Notably, the examples with greatest maximum SCC size are only solved by \toolname{z3}.
\toolname{pray} and \toolname{z3} need at most 95 and 31 seconds, respectively, for the instances where they succeed.
In many cases (e.g., \benchmark{rw-5.01}, \benchmark{golden}, \benchmark{virus}, \benchmark{brown}, \benchmark{swbd}), the resulting certificates formally disprove AST.
For the explicit PPS in \Cref{table:pps}, \toolname{pray} solves all instances whereas \toolname{z3} only solves 3/8 within the time limit, and only finds the trivial solution $\vec 1$.
Most of these benchmarks contain dense high-degree polynomials, and our tool spends most time on performing exact arithmetic.
Standard OVI (rightmost columns in \Cref{table:programs,table:pps}) solves strictly less instances than our eigenvector-based OVI.
On some instances, Standard OVI is slightly faster (if it succeeds).
However, on some larger benchmarks (brown, swbd) our variant runs $\approx 3\times$ faster.

\paragraph{Evaluation of Research Questions.}
(A) Scalability:
Our algorithm succeeds on instances with maximum SCC size of up to 8,000 and number of terms over 50,000.
\toolname{pray} solves all instances with a maximum SCC size of $\leq$ 1,000 in less than 2 minutes per instance.
For the examples where our algorithm does not succeed (e.g., \benchmark{escape100}) it is mostly because it fails converting a floating point to a rational certificate.
(B) PPS with different flavors:
The problems in \Cref{table:programs} (low degree and sparse, i.e., few terms per polynomials) and \Cref{table:pps} (higher degree and dense) are quite different.
A comparison to the SMT approach suggests that our technique might be especially well suited for dense problems with higher degrees.
(C) Non-singularity:
The only instance where our algorithm fails because of the non-singularity condition is the symmetric random walk \benchmark{rw-0.500}.
We therefore conjecture that this condition is often satisfied in practice.
(D) Comparison with baselines:
There is no clear winner.
Some instances can only be solved by one tool or the other (e.g., \benchmark{escape100} and \benchmark{brown}).
However, \toolname{pray} often delivers more succinct certificates, i.e., the rational numbers have less digits.
Moreover, \toolname{z3} behaves \emph{much less predictably} than \toolname{pray}.


%% file: 6_conclusion.tex
\section{Conclusion and Future Work}
\label{sec:conclusion}

We have proposed using inductive bounds as certificates for various properties in probabilistic recursive models, and presented the first dedicated algorithm for computing such bounds.
Our algorithm already scales to non-trivial problems.
A remaining bottleneck is the need for exact rational arithmetic.
This might be improved using appropriate rounding modes as in~\cite{DBLP:conf/tacas/Hartmanns22}.
Additional future work includes certificates for \emph{lower} bounds and termination.

%% file: appendix.tex
\section{Full Proofs}
\label{app:proofs}

\subsection{Proof of~\Cref{thm:existsInductiveChar}}
\label{proof:existsInductiveChar}

\existsInductiveChar*

We now explain the proof of \Cref{thm:existsInductiveChar}.
The proof heavily relies on a \emph{linear approximation} of $\system f$ around the lfp $\lfp\system f$.
Intuitively, this is where the Jacobi matrix $\jac{\system f}(\lfp\system f)$ comes into play.
This is formalized via Taylor's familiar theorem.

\begin{lemma}[{Taylor's Theorem; cf.~\cite[Lem.~2.3]{DBLP:journals/siamcomp/EsparzaKL10}}]
    \label{thm:linearization}
    Let $\system{f}$ be a feasible PPS.
    Then for all vectors $\vec{u} \geq \vec 0$, we have
    \begin{align*}
    \system{f}({\lfp{\system f}} + \vec{u}) 
    \quad=\quad
    \lfp{\system f} ~+~ \jac{\system{f}}(\lfp \system f) \vec{u} ~+~ R_{\vec{u}} \vec{u}
    \end{align*}
    where $R_{\vec{u}}$ is a matrix that depends on $\vec u$ such that $\lim_{\vec{u} \to \vec{0}} R_{\vec u} = 0$.
    More specifically, it holds that $\vec{0} \leq R_{\vec{u}} \vec{u} \leq \big(\jac{\system f}(\lfp \system f + \vec{u}) - \jac{\system f}(\lfp \system f) \big) \vec{u}$.
\end{lemma}

\begin{proof}[Proof of \Cref{thm:existsInductiveChar}]
    \enquote{\ref{it:nonsing} $\implies$ \ref{it:specrad}}:
    By \Cref{thm:specRadBounded} we have $\specRad{\jac{\system f} (\lfp\system f)} \leq 1$.
    Towards contradiction assume that $\specRad{\jac{\system f} (\lfp\system f)} = 1$.
    By the Perron-Frobenius Theorem, $1$ is an eigenvalue of $\jac{\system f} (\lfp\system f)$, which means that there exists $\vec u \neq \vec 0$ such that $\jac{\system f} (\lfp\system f) \vec u = \vec u$.
    This $\vec u$ is in the kernel of $\idMat - \jac{\system f} (\lfp \system{f})$, which contradicts the assumption that $\idMat - \jac{\system f} (\lfp \system{f})$ is non-singular.

    \enquote{\ref{it:specrad} $\implies$ \ref{it:nonsing}}:
    It is a well-known result that for an arbitrary real matrix $M$ the series $\sum_{k = 0}^\infty M^k$ converges iff $\specRad{M} < 1$.
    The limit of the series is the inverse of $\idMat - M$ because
    \[
        (I-M) \sum_{k=0}^\infty M
        ~=~
        \sum_{k=0}^\infty M^k - \sum_{k=1}^\infty M^k
        ~=~
        M^0
        ~=~
        \idMat
        ~.
    \]

    \enquote{\ref{it:specrad} $\implies$ \ref{it:cone}}:
    Let $\specRad{\jac{\system f} (\lfp\system f)} =: \lambda < 1$.
    By the Perron-Frobenius Theorem, the Jacobi matrix $\jac{\system f} (\lfp\system f)$ has a unique normalized eigenvector $\vec{v} \allgreater \vec{0}$ wrt.\ eigenvalue $\lambda$:
    \begin{align}
    \label{eq:eigenvector}
    \jac{\system f} (\lfp\system f) \vec{v}
    ~=~
    \lambda\vec{v}
    ~\allsmaller~
    \vec{v}
    ~.
    \end{align}
    
    Our goal is to define the values $\eps$ and $\delta_{max}$ whose existence we claimed in \Cref{thm:existsInductiveChar}\ref{it:cone}.
    Let $c_{min} > 0$ be the smallest component of $(1-\lambda)\vec{v} \allgreater \vec{0}$.
    We define
    \begin{align}
    \label{eq:defeps}
    \eps ~:=~ \frac{c_{min}}{3 \maxnorm{\jac{\system f}(\lfp\system f)}} ~,
    \end{align}
    where $\maxnorm{\jac{\system f}(\lfp\system f)} = \max_{\maxnorm{\vec y} = 1} \maxnorm{\jac{\system f}(\lfp\system f) \vec y}$ is the maximum row sum of $\jac{\system f}(\lfp\system f)$.
    Note that $\maxnorm{\cdot}$ is the operator norm induced by the maximum norm.
    Then it holds for all $\vec{\eps}$ with $\maxnorm{\vec{\eps}} \leq \eps$ that
    \begin{align}
    \label{eq:2to4intermediate}
    \maxnorm{\jac{\system f}(\lfp\system f)  \vec{\eps}}
    ~\leq~
    \maxnorm{\jac{\system f}(\lfp\system f)} \maxnorm{ \vec{\eps}}
    ~\leq~ 
    \maxnorm{\jac{\system f}(\lfp\system f)} \frac{c_{min}}{3 \maxnorm{\jac{\system f}(\lfp\system f)}}
    ~=~
    \frac{1}{3} c_{min}
    ~.
    \end{align}
    The first inequality in \eqref{eq:2to4intermediate} is a property of operator norms (which is straightforward in the case of the maximum norm).
    Since $c_{min}$ was the smallest component of $(1-\lambda)\vec{v}$, \eqref{eq:2to4intermediate} implies
    \begin{align}
    \label{eq:epslambda}
    \jac{\system f}(\lfp\system f)  \vec{\eps} ~\leq~ \frac{1}{3}(1-\lambda)\vec{v} ~.
    \end{align}
    We now define $\delta_{max}$ as follows:
    \begin{align}
        \label{eq:defDeltaMax}
        \delta_{max} ~:=~ \sup\, \{\delta > 0 \mid \forall \vec \eps \geq \vec 0 \text{ s.t. } \maxnorm{\vec \eps} \leq \eps \colon R_{\delta(\vec{v} + \vec{\eps})} (\vec{v} + \vec{\eps}) \leq \frac{1}{2} (1-\lambda )\vec{v}\} ~,
    \end{align}
    where $R_{\delta(\vec{v} + \vec \eps)}$ is the matrix from \Cref{thm:linearization} which satisfies
    \[
        \system{f}({\lfp{\system f}} + \delta(\vec{v} + \vec \eps))
        ~=~
        \lfp{\system f} + \delta\jac{\system{f}}(\lfp \system f) (\vec{v} + \vec \eps) + \delta R_{\delta(\vec{v} + \vec \eps)} (\vec{v} + \vec \eps) ~.
    \]
    We now argue that $\delta_{max} > 0$.
    This is not immediately obvious because of the $\forall$-quantification in \eqref{eq:defDeltaMax}.
    Let $\delta > 0$ be arbitrary.
    Further, let $\vec\eps \geq \vec 0$ be such that $\maxnorm{\vec \eps} \leq \eps$.
    In the following, we write $\vec \eps' = (\eps \ldots \eps)$.
    We have
    \begin{align*}
        & R_{\delta(\vec{v} + \vec{\eps})} (\vec{v} + \vec{\eps}) \\
        ~=~ & \frac{1}{\delta} R_{\delta(\vec{v} + \vec{\eps})} \delta(\vec{v} + \vec{\eps}) \\
        ~\leq~ & \frac{1}{\delta} \big(\jac{\system f}(\lfp \system f + \delta(\vec{v} + \vec{\eps})) - \jac{\system f}(\lfp \system f) \big) \delta(\vec{v} + \vec{\eps}) \tag{\Cref{thm:linearization}} \\
        ~=~ & \big(\jac{\system f}(\lfp \system f + \delta(\vec{v} + \vec{\eps})) - \jac{\system f}(\lfp \system f) \big) (\vec{v} + \vec{\eps}) \\
        ~\leq~ & \big(\jac{\system f}(\lfp \system f + \delta(\vec v + \vec\eps')) - \jac{\system f}(\lfp \system f) \big) (\vec{v} + \vec\eps') \tag{Jacobi matrix is monotonic}\\
        ~=:~ & M_{\delta} (\vec v + \vec\eps')
    \end{align*}
    Note that $M_{\delta}$ does not depend on $\vec\eps$ and $\lim_{\delta \to 0} M_{\delta} = 0$.
    We can therefore find a specific $\delta^* > 0$ such that $M_{\delta^*}(\vec v + \vec\eps') \leq \frac{1}{2}(1-\lambda)\vec v$.
    On the other hand, we have just shown for all $\vec\eps \geq \vec 0$ with $\maxnorm{\vec \eps} \leq \eps$ and \emph{all} $\delta > 0$ that
    $
        R_{\delta(\vec{v} + \vec{\eps})} (\vec{v} + \vec{\eps}) \leq M_{\delta}(\vec v + \vec\eps')
    $.
    So we have in particular for all $\vec\eps \geq \vec 0$ with $\maxnorm{\vec \eps} \leq \eps$ that
    \[
        R_{\delta^*(\vec{v} + \vec{\eps})} (\vec{v} + \vec{\eps})
        ~\leq~
        M_{\delta^*}(\vec v + \vec\eps')
        ~\leq~
        \frac{1}{2}(1-\lambda)\vec v
        ~.
    \]
    Hence $\delta_{max} \geq \delta^* > 0$.
    
    Finally, let $0 < \delta \leq \delta_{max}$ and $\vec{\tilde v} \geq \vec{v}$ with $\maxnorm{\vec{v} - \vec{\tilde v} } \leq \eps$, i.e., $\vec{\tilde v} = \vec{v} + \vec{\eps}$ for some $\vec{\eps} \geq  \vec 0$ with $\maxnorm{\vec \eps} \leq \eps$.
    Then
    \begin{align*}
    & \system{f}({\lfp{\system f}} + \delta(\vec{v} + \vec \eps))\\
    ~=~& \lfp{\system f} + \delta\jac{\system{f}}(\lfp \system f) (\vec{v} + \vec \eps) + \delta R_{\delta(\vec{v} + \vec \eps)} (\vec{v} + \vec \eps) \tag{by Taylor's Theorem (\Cref{thm:linearization})} \\
    ~=~& \lfp{\system f} + \delta \lambda \vec v + \delta\jac{\system{f}}(\lfp \system f)\vec \eps + \delta R_{\delta(\vec{v} + \vec \eps)} (\vec{v} + \vec \eps) \tag{by \eqref{eq:eigenvector}} \\
    ~\leq~& \lfp{\system f} + \delta \lambda \vec{v} + \delta  \frac{1}{3}(1-\lambda)\vec{v} + \delta R_{\delta(\vec{v} + \vec \eps)} (\vec{v} + \vec \eps) \tag{by \eqref{eq:epslambda}} \\
    ~\leq~& \lfp{\system f} + \delta \lambda \vec{v} + \delta  \frac{1}{3}(1-\lambda)\vec{v} + \delta  \frac{1}{2}(1-\lambda)\vec{v} \tag{by \eqref{eq:defDeltaMax}} \\
    ~\allsmaller~& \lfp{\system f} + \delta \lambda \vec{v} + \delta  \frac{1}{2}(1-\lambda)\vec{v} + \delta  \frac{1}{2}(1-\lambda)\vec{v} \tag{because $\delta (1-\lambda) \vec v \allgreater \vec 0$} \\
    ~=~&  \lfp{\system f} + \delta \vec{v} \tag{simplification}\\
    ~\leq~& \lfp{\system f} + \delta (\vec{v} + \vec \eps) \tag{because $\vec \eps \geq \vec{0}$}\\
    \end{align*}
    
    \enquote{\ref{it:cone} $\implies$ \ref{it:inductiveNotFp}}:
    Trivial.

    \enquote{\ref{it:inductiveNotFp} $\implies$ \ref{it:specrad}}:
    By \ref{it:inductiveNotFp} there exists $\vec u$ such that $\system f(\vec u) < \vec u$.
    By \Cref{thm:induction} this implies that $\lfp\system f < \vec u$, so we can write $\vec u = \lfp\system f +\vec v$ for some $\vec v > \vec 0$.
    
    Using Taylor's Theorem (\Cref{thm:linearization}), it follows that
    \begin{align}
    \label{eq:3to2intermediate}
    \system{f}({\lfp{\system f}} + \vec v) 
    ~=~
    \lfp{\system f} + \jac{\system{f}}(\lfp \system f) \vec v + R_{\vec v} \vec v
    ~<~
    \lfp{\system f} + \vec v
    ~.
    \end{align}
    Using that $R_{\vec v} \vec v \geq \vec 0$, \eqref{eq:3to2intermediate} implies that
    \begin{align}
    \jac{\system{f}}(\lfp \system f) \vec v
    ~<~
    \vec v
    ~.
    \end{align}
    
    
    The claim now follows by applying the following lemma to the matrix $\jac{\system{f}}(\lfp \system f)$ and the vector $\vec v$:
    \begin{lemma}
        \label{thm:MuSmallerU}
        Let $M \geq 0$ be an irreducible $n\times n$-matrix.
        If there exists $\vec u > \vec 0$ such that $M \vec u < \vec u$, then $\vec u \allgreater \vec 0$, $M^n \vec u \allsmaller \vec u$ and $\specRad{M} < 1$.
    \end{lemma}
    \begin{proof}
        First observe that since multiplication by $M$ is monotone we have for all $0 \leq k_1 \leq k_2$ that
        \[
        \vec 0 \leq  M^{k_2} \vec u \leq M^{k_1} \vec u \leq \vec u ~.
        \]
        
        We first show that $\vec u \allgreater \vec 0$, which is essentially \cite[Lemma 5.3]{DBLP:journals/siamcomp/EsparzaKL10}.
        Since $\vec u > \vec 0$, there must be $1 \leq i \leq n$ such that $\vec u_i > 0$.
        Now let $1 \leq j \leq n$ be arbitrary.
        Since $M$ is irreducible there exists $0 \leq k < n$ such that $M^k_{j,i} > 0$.
        This implies that $(M^k \vec u)_j > 0$.
        By monotonicty, $\vec u \geq M^k \vec u$, and thus $\vec u_j \geq (M^k \vec u)_j > 0$.
        Since $j$ was arbitrary, $\vec u \allgreater \vec 0$.
        
        Next we show $M^n \vec u \allsmaller \vec u$.
        Since $M \vec u < \vec u$ holds by assumption, there exists $1 \leq i \leq n$ such that $(M \vec u)_i < \vec u _i$.
        Let $1 \leq j \leq n$ be a arbitrary.
        Since $M$ is irreducible, there exists $0 \leq k < n$ such that $(M^k)_{j,i} > 0$.
        We now show that $(M^n \vec u)_j < u_j$ which implies that $M^n \vec u \allsmaller \vec u$ as $j$ was chosen arbitrarily:
        \begin{align*}
        &(M^n \vec u) _j \\
        ~\leq~ &(M^k M \vec u) _j \tag{by monotonicity, and because $k+1 \leq n$} \\
        ~=~ & (M^k)_{j,i} (M \vec u)_i + \sum_{l\neq i} (M^k)_{j,l} (M \vec u)_l \tag{Def.\ matrix-vector product}\\
        ~<~ & (M^k)_{j,i} \vec u_i + \sum_{l\neq i} (M^k)_{j,l} (M \vec u)_l \tag{because $(M \vec u)_i < \vec u _i$ and $(M^k)_{j,i} > 0$}\\
        ~\leq~ & (M^k)_{j,i} \vec u_i + \sum_{l\neq i} (M^k)_{j,l} \vec u_l \tag{because $(M \vec u)_l \leq \vec u _l$} \\
        ~=~ & (M^k \vec u)_j \leq \vec u_j
        \end{align*}
        
        It remains to show that $\specRad{M} < 1$.
        We do this by showing that the powers of $M$ (i.e., the sequence $(M^k)_{k \geq 0}$) converge to the zero matrix.
        Since $M^n \vec u \allsmaller \vec u$, we can choose $c < 1$ such that $M^n \vec u \leq c \vec u$.
        Then for all $m \geq 1$ it holds that $M^{nm} \vec u \leq c^m \vec u$, so we have
        \[
        \lim_{k \to \infty} M^k \vec u ~=~ \vec 0 ~.
        \]
        Recall from above that we already know $\vec u \allgreater \vec 0$.
        Thus $\lim_{k \to \infty} M^k \vec u = \vec 0$ means that a positive linear combination of the entries of each individual row of $M^k$ converges to zero, i.e., for all $1 \leq i \leq n$ we have $\lim_{k \to \infty} \sum_{j}M^k_{i,j} \vec u_j = 0$, and thus for all $1 \leq j \leq n$, $\lim_{k \to \infty} M^k_{i,j} = 0$.
        Thus $\lim_{k \to \infty} M^k = 0$, which completes the proof.
        %
        %
        \qed
    \end{proof}
    
\end{proof}

\subsection{Proof of~\Cref{thm:oviSoundAndComplete}}
\label{proof:oviSoundAndComplete}

\oviSoundAndComplete*
\begin{proof}
    Correctness is obvious, so we only show termination assuming that $\system f$ is feasible and $\idMat - \jac{\system f}(\lfp\system f)$ is non-singular.
    Clearly, the algorithm terminates iff it eventually finds a $\candidate$ in line 8 which is inductive.
    
    Assume towards contradiction that the algorithm never terminates, i.e., it never finds an inductive $\candidate$.
    For all $i \geq 1$ let $\currLower_i$, $\vec v_i$, $\tol_i$ be the values of the variables $\currLower$, $\vec v$ and $\tol$ at the $i$th time the inner loop at line 7 is reached (note that we then have $N = i-1$).
    Clearly, $\lim_{i\to\infty} \tol_i = 0$.
    By the contract satisfied by $\Update$, we have $\lim_{i\to\infty}\jac{\system f}(\currLower_i) = \jac{\system f}(\lfp\system f)$.
    Since the eigenvectors of $\jac{\system f}(\lfp\system f)$ depend continuously on those of the matrices $\jac{\system f}(\currLower_i)$, and because of the contract satisfied by $\ApproxEV$, the sequence $\vec v_1, \vec v_2, \ldots$ converges to the true unique normalized Perron-Frobenius eigenvector $\vec{v}_{true}$ of $\jac{\system f}(\lfp\system f)$.
    
    We now apply condition \ref{it:cone} of \Cref{thm:existsInductiveChar}.
    The condition ensures that the cone
    \[
    C(\lfp\system f, \vec{v}_{true}, \eps', \delta_{max})
    ~=~
    \{\, \lfp{\system f} + \delta \vec{\tilde v} \mid 0 \leq \delta \leq \delta_{max}, \maxnorm{\vec{\tilde v} - \vec{v}_{true}} \leq \eps' \,\}
    \]
    which is located at $\lfp\system f$, points in direction $\vec{v}_{true}$ and has radius $\eps'$ and length $\delta_{max}$ contains only inductive points.
    For the sake of illustration suppose that the algorithm already knows $\delta_{max}$ and computes $\vec u_i = \currLower_i + \delta \vec v_i$ for some $0 < \delta < \delta_{max}$ instead of executing the loop starting at line 7.
    But then the sequence $(\vec u _i)_{i \geq 1}$ converges to $\lfp\system f + \delta \vec v_{true}$, which is a point that lies \emph{inside the interior} of $C$, so there must be some $i \geq 1$ such that $\vec u_i \in C$, i.e., $\vec u _i$ is inductive.
    
    The remaining difficulty is that $\delta_{max}$ is of course unknown in practice.
    We handle this using the inner loop that starts at line 7.
    Eventually, the variable $N$ is sufficiently large such that $d^k \eps < \delta_{max}$ for some $k \leq N$.
    Termination then follows by applying the argument in the previous paragraph to $\delta = d^k \eps$.
    \qed
\end{proof}

\subsection{Proof of~\Cref{thm:powerIteration}}
\label{proof:powerIteration}

\powerIteration*

\begin{proof}
   Consider the following conditions for an irreducible matrix $M \geq 0$ and a vector $M \vec v_0$ with $M \vec v_0 \neq \vec 0$:
   \begin{enumerate}
       \item $M$ has a unique dominant eigenvalue $|\lambda_1| > |\lambda_2| \geq \ldots \geq |\lambda_n|$.
       \item $\lambda_1$ is \emph{semisimple}, i.e., its algebraic multiplicity\footnote{The algebraic multiplicity is the multiplicity of a given eigenvalue as a root of the characteristic polynomial.} equals its geometric multiplicity\footnote{The geometric multiplicity is the dimension of the eigenspace associated with a particular eigenvalue.}.
       \item $\vec v _0$ is not orthogonal to the eigenspace $\{\vec{v} \mid M \vec v = \lambda_1\vec{v}\}$.
   \end{enumerate}
   It is known that if all these conditions are satisfied, then the power iteration sequence $(\vec v _i)_{i \in \Nats}$ converges to a (normalized) eigenvector $\vec{v}$ with eigenvalue $\lambda_1$ (e.g.~\cite[Theorem 4.1]{saad2011numerical}).
   
   We now show that these conditions are satisfied for the irreducible matrix $M+\idMat \geq 0$ and every initial vector $\vec v_0 > \vec 0$.
   The eigenvectors of $M$ and $M + \idMat$ are exactly the same but the eigenvalues are all shifted by $+1$.
   Indeed, if $\vec v$ is some eigenvector of $M$ with eigenvalue $\lambda$, then $(M+\idMat)\vec v = \lambda \vec v + \vec v = (\lambda + 1)\vec v$.
   However, unlike $M$, the matrix $M + \idMat$ always has period $1$, and so it has a unique dominant eigenvalue $\lambda_1$ by \Cref{thm:perronFrobenius}(2).
   Therefore the first of the above three conditions is satisfied by the matrix $M + \idMat$.
   
   Next, by \Cref{thm:perronFrobenius}(1) it holds that the geometric multiplicity of $\lambda_1$ is 1.
   As the algebraic multiplicity is bounded by the geometric multiplicity, it must also be 1 and thus the matrix $M+\idMat$ satisfies the second condition as well.
   
   Finally, the third condition is satisfied for any $\vec v_0 > \vec 0$ because the scalar product $\vec v_0 \cdot \vec v$ is non-zero (either strictly positive or strictly negative) for all non-zero eigenvectors $\vec v$ of $\lambda_1$ by \Cref{thm:perronFrobenius}(1).
   \qed
\end{proof}

\subsection{Proof of \Cref{thm:basicCertificate}}

\basicCertificate*

\newcommand{\systemRpMod}[1]{\tilde{\system{f}}_{#1}}

\begin{proof}
    This proof closely follows the general idea of decomposed analysis of PPS~\cite{DBLP:conf/stacs/EtessamiY05}.
    
    We first address existence.
    Note that $\systemRp{\ppda}$ is guaranteed to be feasible, in fact $\vec 0 \leq \lfp\systemRp{\ppda} \leq \vec 1$.
    For all $qZr$ with $(\lfp{\systemRp{\ppda}})_{qZr} = 0$ we set $\vec u_{qZr} = 0$.
    By removing these variables from $\systemRp{\ppda}$ we obtain the clean PPS $\systemRpClean{\ppda}$ with $\vec 0 \allsmaller \lfp\systemRpClean{\ppda}$.
    
    Now consider the decomposition of $\systemRpClean{\ppda}$ into the subsystems induced by the strongly connected components of the graph $G_{\systemRpClean{\ppda}}$:
    $\systemRpClean{\ppda}^1,\ldots,\systemRpClean{\ppda}^m$.
    Note that in these subsystems, some variables might only appear on the right hand sides but not on the left (e.g. $x_1 = 0.5x_1 + 0.5x_2, x_2 = 0.5x_1 + 0.5x_3$).
    Since $\lfp\systemRpClean{\ppda} \allgreater \vec 0$, there is a 1 - 1 correspondence of these subsystems and the maximal irreducible submatrices $M_i$ of $\jac{\systemRpClean{\ppda}}(\lfp\systemRpClean{\ppda})$.
    More specifically, $M_i = \jac\systemRpClean{\ppda}^i(\lfp{\systemRpClean{\ppda}})$\footnote{The Jacobi matrix of a sub-PPS with $n' < n$ equations is an $n' \times n'$ matrix where all variables that occur only on the right hand sides are considered constants.}.
    By assumption, $\specRad{M_i} < 1$\footnote{The spectral radius of the zero matrix is zero.}.
    
    Now assume w.l.o.g.\ that $\systemRpClean{\ppda}^1$ is a bottom SCC (i.e., in the dependency graph $G_{\systemRpClean{\ppda}}$ there is no path from the variables in $\systemRpClean{\ppda}^1$ to any variable not in $\systemRpClean{\ppda}^1$).
    Then $\systemRpClean{\ppda}^1$ is a strongly connected PPS with $\jac\systemRpClean{\ppda}^1(\lfp{\systemRpClean{\ppda}}) = \jac\systemRpClean{\ppda}^1(\lfp{\systemRpClean{\ppda}^1})$ and we can apply \Cref{thm:existsInductiveChar}\ref{it:cone} to obtain a rational $\vec u^1$ with $\systemRpClean{\ppda}^1(\vec u^1) \leq \vec u^1$ and $\maxnorm{\lfp\systemRpClean{\ppda}^1 - \vec u^1} \leq \eps$ (in fact, we can do this for any $\eps > 0$).
    
    Suppose we have done the above for all bottom SCCs and now start traversing the DAG of SCCs bottom-up, i.e., in reverse topological order.
    Let $\vec u$ be the bound we have constructed to far (i.e., $\vec u$ contains $\vec u^1$ and the bounds from the other bottom SCC as subvectors and is zero elsewhere).
    Note that we can always make $\vec u$ \emph{smaller} while retaining the inductivity property.
    W.l.o.g. suppose that subsystem $\systemRpClean{\ppda}^2$ is one of the first non-bottom SCCs in the reverse topological order.
    The idea is now to modify $\systemRpClean{\ppda}^2$ to a strongly connected PPS $\systemRpMod{\vec u}^2$ by replacing all variables that occur only in right hand sides by their value in $\vec u$.
    Clearly, $\lim_{\vec u \to \lfp\systemRpClean{\ppda}} \jac\systemRpMod{\vec u}^2(\lfp\systemRpMod{\vec u}^2) = \jac\systemRpClean{\ppda}^2(\lfp{\systemRpClean{\ppda}})$.
    This means we can choose $\vec u$ sufficiently close to $\lfp\systemRpClean{\ppda}$ such that the spectral radius of $\jac\systemRpMod{\vec u}^2(\lfp\systemRpMod{\vec u}^2)$ is strictly smaller than 1.
    We can then apply \Cref{thm:existsInductiveChar}\ref{it:cone} to $\systemRpMod{\vec u}^2$ to obtain a rational $\vec u^2$ with $\systemRpMod{\vec u}^2(\vec u^2) \leq \vec u^2$ to enlarge our current $\vec u$ with.
    
    We can repeat this scheme for all finitely many subsystems until we have constructed a rational $\vec u$ with $\systemRpMod{\vec u}^i(\vec u) \leq \vec u$ for all $i$.
    Clearly, this $\vec u$ also satisfies $\systemRpClean{\ppda}(\vec u) \leq \vec u$.
    Finally, we may extend $\vec u$ by zero entries corresponding to the variables that are assigned zero in the lfp of the (not necessarily clean) $\systemRp{\ppda}$.
    This yields an inductive upper bound for $\systemRp{\ppda}$.
    We stress that in order to verify this bound, we neither have to clean $\systemRp{\ppda}$ nor do we have to compute the SCCs.
    
    For complexity observe that $\systemRp{\ppda}$ is cubic in the size of $\ppda$ and that all polynomials in $\systemRp{\ppda}$ have degree at most 2.
    Since multiplication and addition of rational numbers can be done in polynomial time in the number of their bits, evaluating a polynomial of fixed maximum degree can also be done in polynomial time in the size of the polynomial and the number of bits representing the rationals where the polynomial is to be evaluated.
    Note that this is not true for arbitrary polynomials where exponents are encoded in binary: For instance, evaluating the polynomial $x^{2^n}$ (which can be represented with $\mathcal{O}(n)$ bits) at $x = 2$ yields $2^{2^n}$, a number that needs $\mathcal{O}(2^n)$ bits.
    This means that in order to verify certificates efficiently with exact rational arithmetic, it is important that the polynomials in the PPS do not have very high degrees.
    Fortunately, this is the case for pPDA.
\end{proof}

\clearpage
\section{Certificates for Expected Rewards}
\label{app:pPDA_certification_details}

We can certify upper bounds on the expected value of \emph{rewards} collected during the run of a pPDA.
To simplify the presentation, in this section we assume w.l.o.g. that $\trans{qZ}{p}{r\alpha}$ with $p > 0$ implies $|\alpha| \in \{0,2\}$, i.e., all transitions either decrease or increase the stack height by 1.
Let $\rewFun \colon \ppdaStates \to \NonNegReals$ be a state-based reward function.
Consider the following PPS $\systemRew{\ppda}{\rewFun}$ with variables $\{\exRewVar{q}{Z}{r} \mid qZr \in \ppdaStates \times \abStack \times \ppdaStates\}$:

\begin{align*}
\exRewVar{q}{Z}{r}
&~=~
\sum_{\trans{qZ}{p}{sYX}} p \cdot \sum_{t \in Q} \retProb{s}{Y}{t}  \cdot  \retProb{t}{X}{r}  \cdot  K_{qZ,sYX}
~+~
\sum_{\trans{qZ}{p}{r\eps}} p \cdot \rewFun(r)~,
\end{align*}
where $K_{qZ,sYX} = \rewFun(r) + \exRewVar{s}{Y}{t} + \exRewVar{t}{X}{r}$.
Note that $\systemRew{\ppda}{\rewFun}$ is \emph{linear} but uses the return probabilities which are themselves characterized as the lfp of the \emph{non-linear} system $\systemRt{\ppda}$ from \Cref{thm:retProbEqs} as coefficients.

Suppose that in the lfp $\lfp\systemRew{\ppda}{\rewFun}$, each variable $\exRew{q}{Z}{r}$ is assigned the quantity $\exRew{q}{Z}{r} \in \NonNegRealsInfty$.
It follows from the results of \cite{DBLP:conf/lics/EsparzaKM05} that $\exRew{q}{Z}{r}$ equals the \emph{expected value} of the following random variable $\rvRew^r$ under the probability measure $\prob_\ppda^{qZ}$:
\begin{align*}
\rvRew^r(q_0\gamma_0, q_1 \gamma_1, \ldots)
~=~
\sum_{i > 0}^{\firstHit{r\eps}} \rewFun(q_i) 
\end{align*}
where $\firstHit{r\eps}$ is the minimum integer $k$ such that $q_k \gamma_k = r\eps$, or $0$ if no such $k$ exists.
In words, $\exRew{q}{Z}{r}$ is the expected reward accumulated on the runs from $qZ$ to $r\eps$, where it is assumed that runs which never reach $r\eps$ contribute zero reward.
Consequently, $\exRewTermPaths{q}{Z} = \sum_{r \in \ppdaStates} \exRew{q}{Z}{r}$ is the expected reward accumulated on all terminating runs.


\begin{example}
    Setting $\rewFun = 1$ we can characterize the \emph{expected runtime} of pPDA.
    Reconsider \Cref{ex:irrat_probs}.
    The equation system for expected runtimes becomes
    \begin{align*}
    \exRewVar{q}{Z}{q} = & \frac{1}{4} (\retProb{q}{Z}{q}^2(1 {+} 2\exRewVar{q}{Z}{q}) + \retProb{q}{Z}{r}\retProb{r}{Z}{q}(1 {+} \exRewVar{q}{Z}{r} {+} \exRewVar{r}{Z}{q})  ) + \frac{1}{2}  \\
    \exRewVar{q}{Z}{r} = & \frac{1}{4} (\retProb{q}{Z}{q}\retProb{q}{Z}{r} (1 {+} \exRewVar{q}{Z}{q} {+} \exRewVar{q}{Z}{r}) {+} \retProb{q}{Z}{r}\retProb{r}{Z}{r} (1 {+} \exRewVar{q}{Z}{r} {+} \exRewVar{r}{Z}{r})) + \frac{1}{4}
    \end{align*}
    as well as $\exRewVar{r}{Z}{q} = 0$ and $\exRewVar{r}{Z}{r} = 1$.
    The solution is $\exRewVar{q}{Z}{q} = 2063/2624 \approx 0.786$ and $\exRewVar{q}{Z}{r} = 59/82 \approx 0.712$, so the total expected runtime is $\exRewTermPaths{q}{Z} \approx 1.506$.
    \qedexample
\end{example}

\section{Benchmark Programs}
\label{app:programs}

\begin{figure}
    \centering
    \begin{subfigure}[b]{0.3\textwidth}
        \begin{verbatim}
void f() {
  if flip(p) {
    f();
    f();
  }
}

# main block
{
  f();
}
        \end{verbatim}
        \caption{\benchmark{rw-p}}
    \end{subfigure}
    \hfill
    \begin{subfigure}[b]{0.3\textwidth}
    \begin{verbatim}
void f() {
  if flip(1//2) {
    f();
    f();
    f();
  }
}

# main block
{
  f();
}
    \end{verbatim}
    \caption{\benchmark{golden}}
    \end{subfigure}
    \hfill
    \begin{subfigure}[b]{0.3\textwidth}
        \begin{verbatim}
void offspring() {
  while flip(2//5) {
    offspring();
    while flip(3//5) {
      offspring();
    }
  }
}

# main block
{
  offspring();
}
        \end{verbatim}
        \caption{\benchmark{geom-offspring}}
    \end{subfigure}
    \\
    \vspace{1cm}
    \begin{subfigure}[b]{0.3\textwidth}
        \begin{verbatim}
void gen_operator() {
  uniform(4);
}

void gen_expression() {
  prob {
    4//10: uniform(10);
    3//10: { }
    3//10: {
      gen_operator();
      gen_expression();
      gen_expression();
    }
  }
}

void gen_function() {
  gen_operator();
  gen_expression();
  gen_expression();
}

# main block
{
  gen_function();
}
        \end{verbatim}
        \caption{\benchmark{gun-fun}}
    \end{subfigure}
    \hfill
    \begin{subfigure}[b]{0.3\textwidth}
        \begin{verbatim}
void young() {
  int y = uniform(4);
  while(y > 0) {
    young();
    y = y-1;
  }
  int e = uniform(3);
  while(e > 0) {
    elder();
    e = e-1;
  }
}
void elder() {
  int y = uniform(2);
  while(y > 0) {
   young();
    y = y-1;
  }
  int e = uniform(5);
  while(e > 0) {
    elder();
    e = e-1;
  }
}
# main block
{
  young();
}
        \end{verbatim}
        \caption{\benchmark{virus}}
    \end{subfigure}
    \hfill
    \begin{subfigure}[b]{0.3\textwidth}
        \begin{verbatim}
bool f() {
  prob {
  1//2:
    return flip(1//2);
  1//2:
    if f()  {
      return f();
    } else {
      return false;
    }
  }
}

# main blcok
{
bool res1 = f();
...
bool resN = f();
}       
        \end{verbatim}
        \caption{\benchmark{sequentialN}}
    \end{subfigure}
\end{figure}

\begin{figure}
    \centering
    \begin{subfigure}[b]{0.4\textwidth}
        \begin{verbatim}
int f(int n, int m) {
  prob {
    (n+1)//(n+2) : {
      f((n + 1) % m, m);
      f((n + 1) % m, m);
      return 0;
    }
  1//(n+2) :
    return 0;
  }
}

# main block
{
  f(0, N);
}
        \end{verbatim}
        \caption{\benchmark{escapeN}}
    \end{subfigure}
    \hfill
    \begin{subfigure}[b]{0.4\textwidth}
        \begin{verbatim}
void f(int n) {
  while(n > 0) {
    prob {
      2//3: f(n-1);
      1//3: f((n+1) % N);
    }
    n = n-1;
  }
}

# main block
{
  f(1);
}
        \end{verbatim}
        \caption{\benchmark{modN}}
    \end{subfigure}
\end{figure}

\clearpage
\section{Z3 vs CVC5}
\label{app:z3vscvc5}

\begin{table}[h]
    \caption{
        Comparison of the SMT-approach (see \emph{§Baselines} in \Cref{sec:experiments}) using $\toolname{z3}$ and $\toolname{cvc5}$ on SCFG given as explicit PPS (right), and on programs automatically translated to pPDA (left).
    }
    \label{table:programs}
    \centering
    \setlength{\tabcolsep}{5pt}
    \begin{adjustbox}{max width=\textwidth}
        \rowcolors{1}{}{lightgray}
        \begin{tabular}{l|   c          r    | c      r   }
            \toprule
            benchmark             & $cert_{\toolname{z3}}$ & $t_{\toolname{z3}}$ & $cert_{\toolname{cvc5}}$ & $t_{\toolname{cvc5}}$\\ 
            \midrule
            \benchmark{rw-0.499} &      \cmark   & 11      & \cmark   & 92       \\ 
            \benchmark{rw-0.500}     &  \cmark   & 10      & \cmark   & 87      \\ 
            \benchmark{rw-0.501}     &  \cmark   & 12      & \cmark   & 104       \\ 
            \benchmark{geom-offspring}& \cmark   & 16      & \cmark   & 4687       \\ 
            \benchmark{golden}      &   \cmark   & 14      & \cmark   & 1097       \\ 
            \benchmark{and-or}      &   \cmark   & 15260   & \xmark   & TO       \\ 
            \benchmark{gen-fun}     &   \cmark   & 141     & \xmark   & TO       \\ 
            \benchmark{virus}       &   \cmark   & 139     & \cmark   & 163727       \\ 
            \benchmark{escape10}    &   \cmark   & 48      & \cmark   & 12031       \\ 
            \benchmark{escape25}    &   \cmark   & 15958   & \xmark   & TO     \\ 
            \benchmark{escape50}    &   \cmark   & 410     & \xmark   & TO     \\ 
            \benchmark{escape75}    &   \xmark   & TO      & \xmark   & TO     \\ 
            \benchmark{escape100}    &  \cmark   & 939     & \xmark   & TO     \\ 
            \benchmark{escape200}    &  \xmark   & TO      & \xmark   & TO     \\ 
            \benchmark{sequential5}   & \cmark   & 1074    & \xmark   & TO     \\ 
            \benchmark{sequential7}   & \cmark   & 12822   & \xmark   & TO     \\ 
            \benchmark{sequential10}  & \cmark   & 453718  & \xmark   & TO     \\ 
            \benchmark{mod5}          & \cmark   & 34150   & \xmark   & TO     \\ 
            \benchmark{mod7}          & \cmark   & 443     & \xmark   & TO     \\ 
            \benchmark{mod10}         & \cmark   & 1245    & \xmark   & TO     \\ 
            \bottomrule
        \end{tabular}
        \begin{tabular}{    l      |  c      r        |  c      r       }
            \toprule
            benchmark             & $cert_{\toolname{z3}}$  & $t_{\toolname{z3}}$ & $cert_{\toolname{cvc5}}$ & $t_{\toolname{cvc5}}$\\ 
            \midrule
            \benchmark{brown}      & \xmark      & TO       & \xmark   & TO     \\ 
            \benchmark{lemonde}    & \xmark      & TO       & \xmark   & TO     \\ 
            \benchmark{negra}      & \cmark      & 37248    & \cmark   & 10144     \\ 
            \benchmark{swbd}       & \xmark      & TO       & \xmark   & Error     \\ 
            \benchmark{tiger}      & \cmark      & 17454    & \cmark   & 16118     \\ 
            \benchmark{tuebadz}    & \cmark      & 15323    & \cmark   & 5534     \\ 
            \benchmark{wsj}        & \xmark      & TO       & \xmark   & TO     \\ 
            \benchmark{random}     & \xmark      & TO       & \xmark   & TO     \\ 
            \bottomrule
        \end{tabular}
    \end{adjustbox}
\end{table}

%% file: tacas23-main.bbl
\begin{thebibliography}{10}
\providecommand{\url}[1]{\texttt{#1}}
\providecommand{\urlprefix}{URL }
\providecommand{\doi}[1]{https://doi.org/#1}

\bibitem{DBLP:conf/atva/AzeemEKSW22}
Azeem, M., Evangelidis, A., Kret{\'{\i}}nsk{\'{y}}, J., Slivinskiy, A.,
  Weininger, M.: Optimistic and topological value iteration for simple
  stochastic games. In: {ATVA}. Lecture Notes in Computer Science, vol. 13505,
  pp. 285--302. Springer (2022)

\bibitem{DBLP:books/daglib/0020348}
Baier, C., Katoen, J.: Principles of model checking. {MIT} Press (2008)

\bibitem{DBLP:conf/cav/Baier0L0W17}
Baier, C., Klein, J., Leuschner, L., Parker, D., Wunderlich, S.: {Ensuring the
  Reliability of Your Model Checker: Interval Iteration for Markov Decision
  Processes}. In: {CAV} {(1)}. Lecture Notes in Computer Science, vol. 10426,
  pp. 160--180. Springer (2017)

\bibitem{DBLP:conf/tacas/BarbosaBBKLMMMN22}
Barbosa, H., Barrett, C.W., Brain, M., Kremer, G., Lachnitt, H., Mann, M.,
  Mohamed, A., Mohamed, M., Niemetz, A., N{\"{o}}tzli, A., Ozdemir, A.,
  Preiner, M., Reynolds, A., Sheng, Y., Tinelli, C., Zohar, Y.: cvc5: {A}
  versatile and industrial-strength {SMT} solver. In: {TACAS} {(1)}. Lecture
  Notes in Computer Science, vol. 13243, pp. 415--442. Springer (2022)

\bibitem{kinduction}
Batz, K., Chen, M., Kaminski, B.L., Katoen, J., Matheja, C., Schr{\"{o}}er, P.:
  Latticed k-induction with an application to probabilistic programs. In: {CAV}
  {(2)}. Lecture Notes in Computer Science, vol. 12760, pp. 524--549. Springer
  (2021)

\bibitem{bishop2013model}
Bishop, C.M.: Model-based machine learning. Philosophical Transactions of the
  Royal Society A: Mathematical, Physical and Engineering Sciences
  \textbf{371}(1984),  20120222 (2013)

\bibitem{DBLP:journals/fmsd/BrazdilEKK13}
Br{\'{a}}zdil, T., Esparza, J., Kiefer, S., Kucera, A.: Analyzing probabilistic
  pushdown automata. Formal Methods Syst. Des.  \textbf{43}(2),  124--163
  (2013)

\bibitem{DBLP:journals/jcss/BrazdilKKV15}
Br{\'{a}}zdil, T., Kiefer, S., Kucera, A., Varekov{\'{a}}, I.H.: Runtime
  analysis of probabilistic programs with unbounded recursion. J. Comput. Syst.
  Sci.  \textbf{81}(1),  288--310 (2015)

\bibitem{DBLP:conf/nips/0001R20}
Chiang, D., Riley, D.: {Factor Graph Grammars}. In: NeurIPS (2020)

\bibitem{DBLP:conf/stacs/EsparzaGK10}
Esparza, J., Gaiser, A., Kiefer, S.: {Computing Least Fixed Points of
  Probabilistic Systems of Polynomials}. In: {STACS}. LIPIcs, vol.~5, pp.
  359--370. Schloss Dagstuhl - Leibniz-Zentrum f{\"{u}}r Informatik (2010)

\bibitem{DBLP:conf/stacs/EsparzaKL08}
Esparza, J., Kiefer, S., Luttenberger, M.: {Convergence Thresholds of Newton's
  Method for Monotone Polynomial Equations}. In: {STACS}. LIPIcs, vol.~1, pp.
  289--300. Schloss Dagstuhl - Leibniz-Zentrum f{\"{u}}r Informatik, Germany
  (2008)

\bibitem{DBLP:journals/siamcomp/EsparzaKL10}
Esparza, J., Kiefer, S., Luttenberger, M.: {Computing the Least Fixed Point of
  Positive Polynomial Systems}. {SIAM} J. Comput.  \textbf{39}(6),  2282--2335
  (2010)

\bibitem{DBLP:conf/lics/EsparzaKM04}
Esparza, J., Kucera, A., Mayr, R.: {Model Checking Probabilistic Pushdown
  Automata}. In: {LICS}. pp. 12--21. {IEEE} Computer Society (2004)

\bibitem{DBLP:conf/lics/EsparzaKM05}
Esparza, J., Kucera, A., Mayr, R.: {Quantitative Analysis of Probabilistic
  Pushdown Automata: Expectations and Variances}. In: {LICS}. pp. 117--126.
  {IEEE} Computer Society (2005)

\bibitem{DBLP:conf/cav/EsparzaLNNSS13}
Esparza, J., Lammich, P., Neumann, R., Nipkow, T., Schimpf, A., Smaus, J.: A
  fully verified executable {LTL} model checker. In: {CAV}. Lecture Notes in
  Computer Science, vol.~8044, pp. 463--478. Springer (2013)

\bibitem{DBLP:conf/stacs/EtessamiY05}
Etessami, K., Yannakakis, M.: {Recursive Markov Chains, Stochastic Grammars,
  and Monotone Systems of Nonlinear Equations}. In: {STACS}. Lecture Notes in
  Computer Science, vol.~3404, pp. 340--352. Springer (2005)

\bibitem{DBLP:journals/jacm/EtessamiY09}
Etessami, K., Yannakakis, M.: {Recursive Markov chains, stochastic grammars,
  and monotone systems of nonlinear equations}. J. {ACM}  \textbf{56}(1),
  1:1--1:66 (2009)

\bibitem{DBLP:books/daglib/0023751}
Flajolet, P., Sedgewick, R.: {Analytic Combinatorics}. Cambridge University
  Press (2009)

\bibitem{DBLP:conf/tacas/FunkeJB20}
Funke, F., Jantsch, S., Baier, C.: {Farkas Certificates and Minimal Witnesses
  for Probabilistic Reachability Constraints}. In: {TACAS} {(1)}. Lecture Notes
  in Computer Science, vol. 12078, pp. 324--345. Springer (2020)

\bibitem{DBLP:conf/rp/HaddadM14}
Haddad, S., Monmege, B.: Reachability in mdps: Refining convergence of value
  iteration. In: {RP}. Lecture Notes in Computer Science, vol.~8762, pp.
  125--137. Springer (2014)

\bibitem{DBLP:conf/tacas/Hartmanns22}
Hartmanns, A.: {Correct Probabilistic Model Checking with Floating-Point
  Arithmetic}. In: {TACAS} {(2)}. Lecture Notes in Computer Science, vol.
  13244, pp. 41--59. Springer (2022)

\bibitem{ovi}
Hartmanns, A., Kaminski, B.L.: Optimistic value iteration. In: {CAV} {(2)}.
  Lecture Notes in Computer Science, vol. 12225, pp. 488--511. Springer (2020)

\bibitem{DBLP:phd/dnb/Jantsch22}
Jantsch, S.: {Certificates and Witnesses for Probabilistic Model Checking}.
  Ph.D. thesis, Dresden University of Technology, Germany (2022)

\bibitem{DBLP:conf/atva/Jantsch0B20}
Jantsch, S., Funke, F., Baier, C.: {Minimal Witnesses for Probabilistic Timed
  Automata}. In: {ATVA}. Lecture Notes in Computer Science, vol. 12302, pp.
  501--517. Springer (2020)

\bibitem{DBLP:conf/icassp/JurafskyWSSFTM95}
Jurafsky, D., Wooters, C., Segal, J., Stolcke, A., Fosler, E., Tajchman, G.N.,
  Morgan, N.: Using a stochastic context-free grammar as a language model for
  speech recognition. In: {ICASSP}. pp. 189--192. {IEEE} Computer Society
  (1995)

\bibitem{DBLP:journals/dam/Karp91}
Karp, R.M.: An introduction to randomized algorithms. Discret. Appl. Math.
  \textbf{34}(1-3),  165--201 (1991)

\bibitem{DBLP:conf/stoc/KieferLE07}
Kiefer, S., Luttenberger, M., Esparza, J.: On the convergence of newton's
  method for monotone systems of polynomial equations. In: {STOC}. pp.
  217--226. {ACM} (2007)

\bibitem{DBLP:journals/nar/KnudsenH03}
Knudsen, B., Hein, J.: Pfold: {RNA} secondary structure prediction using
  stochastic context-free grammars. Nucleic Acids Res.  \textbf{31}(13),
  3423--3428 (2003)

\bibitem{DBLP:journals/lmcs/KobayashiLG20}
Kobayashi, N., {Dal Lago}, U., Grellois, C.: {On the Termination Problem for
  Probabilistic Higher-Order Recursive Programs}. Log. Methods Comput. Sci.
  \textbf{16}(4) (2020)

\bibitem{DBLP:journals/lmcs/KuceraEM06}
Kucera, A., Esparza, J., Mayr, R.: Model checking probabilistic pushdown
  automata. Log. Methods Comput. Sci.  \textbf{2}(1) (2006)

\bibitem{DBLP:journals/csr/McConnellMNS11}
McConnell, R.M., Mehlhorn, K., N{\"{a}}her, S., Schweitzer, P.: Certifying
  algorithms. Comput. Sci. Rev.  \textbf{5}(2),  119--161 (2011)

\bibitem{DBLP:journals/corr/abs-1809-10756}
van~de Meent, J., Paige, B., Yang, H., Wood, F.: {An Introduction to
  Probabilistic Programming}. CoRR  \textbf{abs/1809.10756} (2018)

\bibitem{jgrapht}
Michail, D., Kinable, J., Naveh, B., Sichi, J.V.: Jgrapht - {A} java library
  for graph data structures and algorithms. {ACM} Trans. Math. Softw.
  \textbf{46}(2),  16:1--16:29 (2020)

\bibitem{DBLP:conf/tacas/MouraB08}
de~Moura, L.M., Bj{\o}rner, N.S.: {Z3:} an efficient {SMT} solver. In: {TACAS}.
  Lecture Notes in Computer Science, vol.~4963, pp. 337--340. Springer (2008)

\bibitem{DBLP:conf/lics/OlmedoKKM16}
Olmedo, F., Kaminski, B.L., Katoen, J., Matheja, C.: {Reasoning about Recursive
  Probabilistic Programs}. In: {LICS}. pp. 672--681. {ACM} (2016)

\bibitem{DBLP:conf/cav/QuatmannK18}
Quatmann, T., Katoen, J.: {Sound Value Iteration}. In: {CAV} {(1)}. Lecture
  Notes in Computer Science, vol. 10981, pp. 643--661. Springer (2018)

\bibitem{nonnegmats}
Rothblum, U.G.: Nonnegative matrices and stochastic matrices. In: Hogben, L.
  (ed.) Handbook of {L}inear {A}lgebra. CRC press (2006)

\bibitem{saad2011numerical}
Saad, Y.: Numerical methods for large eigenvalue problems: revised edition.
  SIAM (2011)

\bibitem{DBLP:journals/prl/SimistiraKC15}
Simistira, F., Katsouros, V., Carayannis, G.: Recognition of online handwritten
  mathematical formulas using probabilistic {SVM}s and stochastic context free
  grammars. Pattern Recognit. Lett.  \textbf{53},  85--92 (2015)

\bibitem{DBLP:journals/jacm/StewartEY15}
Stewart, A., Etessami, K., Yannakakis, M.: {Upper Bounds for Newton's Method on
  Monotone Polynomial Systems, and P-Time Model Checking of Probabilistic
  One-Counter Automata}. J. {ACM}  \textbf{62}(4),  30:1--30:33 (2015)

\bibitem{DBLP:conf/tacas/WimmerM20}
Wimmer, S., von Mutius, J.: {Verified Certification of Reachability Checking
  for Timed Automata}. In: {TACAS} {(1)}. Lecture Notes in Computer Science,
  vol. 12078, pp. 425--443. Springer (2020)

\bibitem{DBLP:conf/fossacs/WinklerGK22}
Winkler, T., Gehnen, C., Katoen, J.: {Model Checking Temporal Properties of
  Recursive Probabilistic Programs}. In: FoSSaCS. Lecture Notes in Computer
  Science, vol. 13242, pp. 449--469. Springer (2022)

\bibitem{zenodo}
Winkler, T., Katoen, J.: {Artifcat Evaluation for TACAS '23 Paper:
  ``Certificates for Probabilistic Pushdown Automata via Optimistic Value
  Iteration''} (2023). \doi{10.5281/zenodo.7506305}

\bibitem{arxiv}
Winkler, T., Katoen, J.P.: {Certificates for Probabilistic Pushdown Automata
  via Optimistic Value Iteration}  (2023). \doi{10.48550/ARXIV.2301.08657}

\bibitem{DBLP:phd/ethos/Wojtczak09}
Wojtczak, D.: Recursive probabilistic models : efficient analysis and
  implementation. Ph.D. thesis, University of Edinburgh, {UK} (2009)

\bibitem{premo}
Wojtczak, D., Etessami, K.: {PReMo : An Analyzer for Probabilistic Recursive
  Models}. In: {TACAS}. Lecture Notes in Computer Science, vol.~4424, pp.
  66--71. Springer (2007)

\bibitem{DBLP:conf/qest/YannakakisE05}
Yannakakis, M., Etessami, K.: Checking {LTL} properties of recursive markov
  chains. In: {QEST}. pp. 155--165. {IEEE} Computer Society (2005)

\end{thebibliography}
